\documentclass[12pt, draftclsnofoot, onecolumn]{IEEEtran}

\usepackage{amsmath,amssymb,mathtools}
\usepackage{dsfont}
\usepackage{mathrsfs,amsfonts}
\usepackage{multirow}
\usepackage{amsthm}
\usepackage{cleveref}

\usepackage{color}
\usepackage{bbold}

\usepackage{slashbox}
\usepackage{longtable}

\usepackage{subcaption}
\usepackage[noadjust]{cite}
\usepackage{url}

\newtheorem{MomentsLemma}{Lemma}
\newtheorem{Theorem}{Theorem}

\newtheorem{Proposition}{Proposition}
\newtheorem{corollary}{Corollary}
\newtheorem{remark}{Remark}

\theoremstyle{plain}

\ifCLASSINFOpdf
   \usepackage[pdftex]{graphicx}
\else
 \usepackage[dvips]{graphicx}
  \graphicspath{{../}}
\fi
\usepackage{epstopdf}

\usepackage[T1]{fontenc}
\usepackage{array}
\usepackage{makecell}
\newcolumntype{x}[1]{>{\centering\arraybackslash}p{#1}}
\usepackage{tikz}
\newcommand\diag[4]{%
  \multicolumn{1}{p{#2}|}{\hskip-\tabcolsep
  $\vcenter{\begin{tikzpicture}[baseline=0,anchor=south west,inner sep=#1]
  \path[use as bounding box] (0,0) rectangle (#2+2\tabcolsep,\baselineskip);
  \node[minimum width={#2+2\tabcolsep-\pgflinewidth},
        minimum  height=\baselineskip+\extrarowheight-\pgflinewidth] (box) {};
  \draw[line cap=round] (box.north west) -- (box.south east);
  \node[anchor=south west] at (box.south west) {#3};
  \node[anchor=north east] at (box.north east) {#4};
 \end{tikzpicture}}$\hskip-\tabcolsep}}

\usepackage{lipsum}             
\usepackage{xargs}
\usepackage[disable,colorinlistoftodos,prependcaption,textsize=tiny,textwidth=0.85in]{todonotes} 
\newcommandx{\change}[2][1=]{\todo[linecolor=blue,backgroundcolor=blue!20,bordercolor=blue,#1]{#2}}
\newcommandx{\info}[2][1=]{\todo[linecolor=red,backgroundcolor=red!25,bordercolor=red,#1]{#2}}

\newcommand{\norm}[1]{\left\lVert#1\right\rVert}

\makeatletter
\newcommand{\vast}{\bBigg@{3}}
\newcommand{\Vast}{\bBigg@{5}}
\makeatother

\begin{document}

\title{Interference Modeling for Cellular Networks under Beamforming Transmission}
\author{{\normalsize \IEEEauthorblockN{Hussain Elkotby\thanks{\!\!\!\!\!\!\!\! 
The authors are with the Department of Electrical and Computer Engineering,
Tufts University,
Medford, MA, USA. Emails: Hussain.Elkotby@Tufts.edu, Mai.Vu@Tufts.edu}, \textit{Student Member IEEE} and Mai Vu, \textit{Senior Member IEEE}}\\
}}

\maketitle

\vspace{-12mm}
\begin{abstract}
\vspace{-.15in}
We propose analytical models for the interference power distribution in a cellular system employing MIMO beamforming in {rich and limited scattering environments, which capture {non line-of-sight } signal propagation in the microwave and mmWave bands, respectively.} Two candidate models are considered: the Inverse Gaussian and the Inverse Weibull, both are two-parameter heavy tail distributions. We further propose a mixture of these two distributions as a model with three parameters. To estimate the parameters of these distributions, three approaches are used: moment matching, individual distribution maximum likelihood estimation (MLE), and mixture distribution MLE with a designed expectation maximization algorithm. We then introduce simple fitted functions for the mixture model parameters as polynomials of {the channel path loss exponent and shadowing variance}. To measure the goodness of these models, the information-theoretic metric relative entropy is used to capture the distance from the model distribution to a reference one. {The interference models are tested against data obtained by simulating a cellular network based on stochastic geometry. The results show that the three-parameter mixture model offers remarkably good fit to simulated interference power.} The mixture model is further used to analyze the capacity of a cellular network employing joint transmit and receive beamforming and confirms a good fit with simulation.

\textit{Keywords: mmWave cellular; interference model; stochastic geometry; moment matching; maximum likelihood estimation; mixture distribution.}
\end{abstract}

\IEEEpeerreviewmaketitle

\vspace{-.17in}
\section{Introduction}\label{intro}\vspace{-.1in}
\IEEEPARstart{T}{he} current scarcity of wireless spectrum coupled with the predicted exponential increase in capacity demand has focused attention on denser network deployment and the underutilized {millimeter wave (mmWave)} bands (30-300 GHz) \cite{6398884,mmWaveEnable,ref_2e}. 
 Fifth generation (5G) wireless systems aiming to cater for the capacity demand are expected to provide a minimum of 1 Gb/s data rate anywhere with up to 5 Gb/s for high mobility users and 50 Gb/s data rates for pedestrians \cite{mmWaveEnable}. 
%
%
 Both dense deployments and mmWave make promising candidates for the 5G systems through the resources reuse over smaller areas and the huge amount of available spectrum \cite{Baimm14,mmWaveEnable}.  Modeling and characterization of wireless interference under these scenarios is essential for cellular system analysis and design. Due to the different characteristics of signal propagation in conventional microwave and mmWave bands, we use the terminology rich and limited scattering environments to refer to the signal propagation in microwave and mmWave bands.

\vspace{-.15in}
\subsection{Background and Related Works}\vspace{-.1in}
We argue for the important role that interference characterization plays in evaluating and predicting the network performance in \info{, including} both microwave and mmWave bands. 
Traditionally, mmWave bands are considered for backhaul in cellular systems and for high-volume
consumer electronics such as personal area and local
area networks, but not for cellular access due to concerns about short-range
and non-line-of-sight coverage issues \cite{Baimm14,mmWaveEnable}. MmWave has, however, recently been shown {to be} suitable for
cellular communications, provided short cell radius of the order of 100-200 meters and sufficient beamforming gain between communicating nodes \cite{mmWaveEnable}. Reducing the cell radius leads to dense base station {(BS)} deployments. Even under beamforming, these high BS and user densities can drive cellular networks to be more interference rather than noise limited.  
While large adaptive arrays with narrow beams can boost the received signal power and hence reduce the impact of out-of-cell interference \cite{6894455,ref_2e}, this interference remains an important performance-limiting factor in dense mmWave networks \cite{Int_Regime_7499308}.

In next generation of wireless networks with a large number of subscribers and dense BSs deployment, interference modeling is an important step towards {network} realization. 
Existing {stochastic geometry based} work either uses the {interference} Laplace transform to evaluate simple transmission schemes \cite{ref3, ref_6}, or models the interference 
using moment matching gamma distribution \cite{HElsawy6524460,heathHetero6515339}. 
 Gaussian distribution is another approximation that is considered for the centered and normalized aggregate wireless interference power \cite{Gaussian_Int_7541571,Gauss_Int_5549941}. The central limit theorem which justifies the Gaussian approximation, however, does not apply when some of the interferers are dominant. Also, the Gaussian distribution does not model the interference very well at low density of interferers or when the exclusion region, the region with no interferers, is relatively small, as the cell sizes shrink. The distribution of the interference power at relatively small exclusion regions has a heavy tail which can not be captured by the Gaussian distribution \cite{Gauss_Int_5549941}.
Here we propose 
analytical distribution models characterized by only a few parameters, which can be fitted to simple polynomial functions of channel path loss exponent and shadowing variance, and test their fitness against network simulation based on stochastic geometry. 
A closed form distribution of the interference helps in designing a cellular system by allowing the analysis of system capacity and comparing different transmission techniques, which may not be directly achievable with a characteristic function of the interference using the Laplace transform. Further, having the interference distribution can significantly speed up system capacity evaluation by generating the interference directly to use in capacity numerical computation, instead of running time-consuming simulations.

To evaluate an interference model, we apply it to a cellular network and analyze the system performance. For cellular network analysis, stochastic geometry 
{is shown to capture the main performance trends with analytical tractability}. 
{Stochastic geometry} is used to develop a tractable model for downlink heterogeneous cellular networks 
\cite{ref_6}, which is applied to analyze coordinated multipoint beamforming \cite{ref3}.
These stochastic geometry based networks are useful in verifying that the performance of a cellular network 
matches experimental trends, and can also be used to verify the accuracy of an interference model.
\vspace{-.2in}
\subsection{Approaches to Interference Modeling}\vspace{-.1in}
The underlying question is "\textit{What parameterized distribution can best model out-of-cell interference and how can we estimate the parameters of this distribution?}". We focus on parameterized distributions with as few parameters as possible to make the model the simplest while having a good fit. 
We test the goodness of the proposed interference models against simulation of a cellular network based on stochastic geometry. We note that testing against actual measurement data is also feasible and is desirable when such data are available. 

To estimate the parameters of the proposed interference models, we use both moment matching and maximum likelihood estimation (MLE) techniques. 
Specifically, we consider three approaches: analytical moment matching, individual distribution MLE, and mixture distribution MLE. 
 We exploit the iterative expectation maximization (EM) algorithm to estimate the parameters of the mixture MLE model. To evaluate the goodness of each model, we introduce the use of the information-theoretic relative entropy or Kullback-Leibler (KL) divergence as a measure for the relative distance between the modeled distributions and simulated data. 

In our interference model development, we consider the interference at a {BS} in a MIMO uplink scenario, but the results are also applicable to downlink. 
The developed interference models apply to all transmit beamforming techniques, as long as the beamforming vector is independent of the interference channels.
We verify the modeled interference power distribution against network simulation for a wide range of realistic {channel} propagation parameters.
\vspace{-.2in}
\subsection{Main Results and Contributions}\vspace{-.1in}
We develop out-of-cell interference models that represent the interference power 
as a random variable with a known, parametrized probability density function. This representation is important for cellular network performance evaluation, prediction and design. We further apply the developed models in analyzing the capacity and outage performance of a MIMO cellular network employing joint transmit and receive beamforming. 

The main contributions and novelties of this paper are a new analytical model for interference power distribution and methods for estimating its parameters for both rich scattering and limited scattering environments, which are summarized as:

\begin{enumerate}
\item[1)] We propose the use of two distributions to model interference power, each distribution characterized by two parameters: the {inverse Gaussian} (IG) as a light-to-heavy tailed distribution and the {inverse Weibull} (IW) as a heavy tailed {one}. Further, we propose a novel model as a mixture of these two distributions with remarkably good fit to simulation data while having only three  parameters.

\item[2)] We apply three approaches to estimate the interference models parameters: moment matching (MM), individual distribution MLE, and mixture distribution MLE. In the MM approach, a simple matching of the first two 
interference power moments is used. 
In the other two approaches, a combination of MM and MLE techniques is used in designing an iterative EM algorithm which maximizes a log likelihood function of the interference data.

\item[3)] We propose the use of the relative entropy or Kullback-Leibler distance from information theory to measure the goodness of each model. This metric measures the relative distance between a modeled distribution and reference (simulated) data, which gives a good indication of how far the proposed interference model is from the referenced interference.

\item[4)] We provide simple polynomial functions with fitted coefficients to express the mixture MLE model parameters in terms of channel characteristics, including the path loss exponent and shadowing standard deviation. These polynomials can be used as a simple representation of network interference in complex system-level simulations.

\item[5)] We apply the interference models to evaluate the performance of a cellular network with joint transmit and receive \info{dominant mode} beamforming. User's outage probabilities using the MM and mixture MLE interference models are compared to stochastic geometry based system simulation. 

\item[6)] Our proposed mixture model shows excellent fit in both interference distribution and network performance evaluation for a wide range of {channel} propagation parameters, including path loss exponent from 2 to 5 and shadowing standard deviation from 0 dB to 9 dB. 
\end{enumerate}

This work can be extended to model interference in more complicated propagation environments, including those with a probabilistically parametrized {LOS} and {NLOS} channel model, and to more heterogeneous network deployment such as those involving relaying nodes \cite{BaiH14,MyArxiv,ref_1r}. 

The importance of our developed mixture model comes from the fact that we can fit the data set of out-of-cell interference into a simple mathematical model, a distribution with known probability density function. This mathematical model can be used to study the performance of a cellular network where model parameters can be specifically tailored for each network setting. These parameters can be derived by applying the MLE method to interference data obtained either via simulation (such as those based on stochastic geometry networks, including point processes and random shapes), or via actual measurement campaigns. The flexibility of adapting the parameters of this mixture distribution makes it a versatile tool in modeling interference.

\vspace{-.1in} 
\section{Models for network and channel propagation}\label{Sec1}
\vspace{-.1in}
\subsection{Network model based on stochastic geometry} \label{GeoModel}\vspace{-.1in}
We consider a cellular system consisting of multiple cells with an average cell radius $R_0$. 
Each cell has a single BS that is equipped with $N_{\text{BS}}$ antennas and serves multiple user equipments (UEs). Each UE is equipped with $N_{\text{UE}}$ antennas and uses a distinct resource block in each cell, hence there is no intra-cell interference. However, each UE suffers from  out-of-cell interference due to frequency reuse in all other cells. In this paper, we denote all transmissions on the same resource block from all cells other than the cell under study as the out-of-cell interference.   
 
 Due to the irregular structure of current and future cellular networks, we employ stochastic geometry to describe the network. We consider uplink transmissions in this paper, although results are also applicable to downlink. We model the active UEs in different cells contending for the same resource block and causing interference to each other as being distributed on a two-dimensional plane according to a homogeneous and stationary Poisson point process {(PPP)} $ \Phi_1 $ with intensity $ \tilde{\lambda}_1 = \eta \lambda_1 ${, where $\eta$ represents the user density factor and can be varied}. 
Furthermore, under the assumption that each BS serves a single mobile in a given resource block, we adopt the model in \cite{ref2} which places each BS uniformly in the 
Voronoi cell of its served UE. According to this model, the distance, $r_c$, between each BS and its served UE is Rayleigh distributed and has an average value of $R_c=1/({2\sqrt{\eta\lambda_1}})$. Throughout this paper, we assume that $\lambda_1$ is fixed and that $R_c$ changes with the parameter $\eta$.

 In order to develop a model for the out-of-cell interference power in a stochastic geometry based network, we consider a cell under study with a typical radius $R_c$ in a field of active UE interferers with intensity $\tilde{\lambda}_1$. 
 As such, our network model consists of a typical cell under study centered at its BS and surrounded by Voronoi cells of other active uplink UEs who interfere with the considered BS at the origin. Our goal is to model this out-of-cell interference to the considered BS receiver, taking into account channel propagation {features} as discussed next.

\vspace{-.1in}
\subsection{Channel propagation model}\label{PhyCh}\vspace{-.1in}
We consider two channel models for two different environments, the rich and the limited scattering environments, which are typical for signal propagation in the microwave and mmWave bands, respectively.
For the rich scattering environment, we consider a complete channel model with shadowing, path loss and small scale fading. 
As such, we express a typical MIMO channel with Tx-Rx distance $r$ in the following form: \vspace{-.15in}
\begin{eqnarray}
\mathbf{H}=\sqrt{\mathit{l}(r)} \mathbf{\tilde{H}}
\end{eqnarray}
where $\mathbf{\tilde{H}}$ is a random matrix that captures the effects of small scale fading and which can be modeled as i.i.d. $\mathcal{CN}({0},1) $; and function $\mathit{l}(r)$ captures the large scale fading which includes both path loss and shadowing. 
The large scale fading function $\mathit{l}(r)$ is modeled as a log-normal random variable multiplied with the pathloss as a function of the distance $r$ as follows: \vspace{-.1in}
\begin{eqnarray}\label{PathModel}
\mathit{l}(r)&=&\mathit{L}_s \beta r^{-\alpha} = \mathit{L}_s \mathit{l}_p(r),
\end{eqnarray} 
where {$\mathit{L}_{s} \sim log(0,\sigma_{SF})$} is a log-normal random variable with standard deviation {$\sigma_{SF}$} dB; $\alpha$ is the pathloss exponent; and $\beta$ is the intercept of the pathloss formula. The intercept represents the reference attenuation point that determines the tilt of the path loss model \cite{maccartney2013path}.

For the limited scattering environment, as typical in mmWave,  the signal propagation characteristics differ \info{differs} considerably. The path loss in mmWave systems is severe with distance-dependent LOS and NLOS propagation, often modeled as a probabilistic function of mixed LOS and NLOS \cite{Baimm14, BaiH14}. Further, mmWave signals can be severely vulnerable to shadowing and blockage effects \cite{mmWave4}. For limited scattering, however, we only focus on the NLOS interference component in this paper, leaving the composite interference from both NLOS and LOS components as a future work. As such, the large scale fading can be modeled as having a single slope as in \eqref{PathModel}.

The main feature, due to limited scattering, is that the channel becomes highly directional. Specifically, for the small scale fading, a directional $\mathbf{\tilde{H}}$ can be defined for $K$ scatterers 
as \vspace{-.08in}
\begin{eqnarray}\label{Full_model_ch}
\mathbf{\tilde{H}} = \sum_{k_1=1}^K a_{k_1} \mathbf{u}_{rx} \left( \theta_{k_1}^{rx} \right)   \mathbf{u}_{tx}^{\ast} \left( \theta_{k_1}^{tx} \right)
\end{eqnarray}
where $K$ denotes the number of channel path clusters; $a_{k_1}$ is the complex channel coefficient of the $k_1^{th}$ path cluster, assuming (2-D) beamforming; $ \theta_{k_1}^{rx}$ and $ \theta_{k_1}^{tx}$ respectively are the horizontal angles of arrival (AoA) and departure (AoD) of the $k_1^{th}$ path cluster; and $\mathbf{u}_{rx}(\cdot) \in \mathbb{C}^{N_{\text{BS}}}$ and $\mathbf{u}_{tx}\left(\cdot\right) \in \mathbb{C}^{N_{\text{UE}}}$ are the vector response functions for the BS and UE antenna arrays.

Based on {these} channel models, we can express the channels of the direct link between the considered active UE and its BS and the interfering links from other active UEs as follows: \vspace{-.08in}
\begin{eqnarray}
\!\!\!\!\!\!\!\!\!\!\!\!\!\!{\mathbf{H}_{0}^{}}&=&{} \sqrt{\mathit{l}\left({\Vert \bold{p}_0 \Vert}_2\right)}\mathbf{\tilde{H}}_{0},\quad \quad \quad \;\;
{\mathbf{H}_{k}}={} \sqrt{\mathit{l}\left({\Vert \bold{z}_k \Vert}_2\right)}\mathbf{\tilde{H}}_{k},
\label{eq2} 
\end{eqnarray}
\!\!\!\! where $ \bold{p}_0 $ and $ \bold{z}_k $ are vectors representing the 2-$ D $ location of the considered active UE and the $k^{th}$ interfering UE in $ \Phi_1 $ with respect to the origin (i.e. the considered BS). Here $ \mathbf{H}_{0} $ is the direct channel from the considered active UE and $ \mathbf{H}_{k} $ is the channel from the $ k^{th} $ interfering UE in $ \Phi_1 $ to the considered BS.

Through a number of recent mmWave channel measurement campaigns, channel parameters for the model in \eqref{PathModel} have been identified \cite{mmWave2,mmWave4,mmWave6515173,mmWave6387266}. For channel propagation at 28 GHz, typical values of the path loss exponent range from 1.68 to 2.55 in {LOS} and from 2.92 
to 5.76 in {NLOS} environments. {Shadowing} standard deviation ranges from 0.2 dB to 8.66 dB in {LOS} and from 8.7 dB to 9.7 dB 
in {NLOS}. In this paper, we develop {an interference} model that works well for these ranges of channel parameters and beyond, as shown later in the numerical section.

\vspace{-.1in} 
\section{Interference Formulation under MIMO Beamforming}\label{RelayScheme}\vspace{-.1in}
In this section, we present the MIMO beamforming signal model and formulate the out-of-cell interference at the considered BS under transmit beamforming from all active UEs in rich and limited scattering environments. Beamforming has been adopted as an essential technique for mmWave communications to overcome the huge propagation loss at high mmWave carrier frequencies \cite{mmWaveEnable,mmWaveSurvey15}.  We then formulate the per-user achievable rate and outage probability under an example of dominant mode joint transmit and receive beamforming in the rich scattering environment, and directional analog beamforming in the limited scattering environment.
\vspace{-.1in}
\subsection{MIMO Beamforming Signal Model}\label{original model}\vspace{-.1in}
We now describe the signal model for {single-stream MIMO beamforming}, i.e., no spatial multiplexing.
We model the received signal $\{
\textbf{y}_{0} \in \mathbb{C}^{N_{\text{BS}}\times 1}\}$ at the considered {BS} as \vspace{-.1in}
\begin{align}\label{Eq_4r1}
\textbf{y}_{0} &= \textbf{H}_{0} \textbf{x}_{0} +\textbf{v}_{0}+\textbf{z}_{0}
\end{align}
where $\textbf{x}_{0} \in \mathbb{C}^{N_{\text{UE}}\times 1}$ is the transmitted signal vector from the considered UE; $\textbf{z}_{0} \in \mathbb{C}^{N_{\text{BS}}\times 1}$ is an i.i.d. \info{i.i.d} noise vector distributed as $ \mathcal{CN}(\textbf{0},\sigma^2 I) $; $\textbf{H}_{0}\in \mathbb{C}^{N_{\text{BS}}\times N_{\text{UE}}}$ is the considered UE-to-BS channel matrix; and $\textbf{v}_{0} \in \mathbb{C}^{N_{\text{BS}}\times 1}$ represents the interference vector received at the considered BS from interfering UEs in all other cells.
 
The transmit signal vector under beamforming from each user can be generally described as \vspace{-.1in}
\begin{eqnarray}\label{B3}
\textbf{x}_{}&\!\!\!=&\!\!\! \textbf{w}_{} \sqrt{P_{}} U_{}.
\end{eqnarray} 
where $U_{}$ is a standard Gaussian signal with zero mean and unit variance; $P_{}$ represents the total power allocated to the active {UE} within a single transmission period; and $\textbf{w}_{} \in \mathbb{C}^{N_{\text{UE}}\times 1}$ is the unit norm beaforming vector at the active {UE}. We will use this signal model with the appropriate subscript to denote the transmit vector from an active UE, either intended or interfering one.

\vspace{-.1in}
\subsection{Interference {\color{black}Signal} Formulation}\vspace{-.1in}
For the purpose of modeling network-wide interference, we emphasize that the distribution of interference is independent of the beamforming scheme employed. The only condition is that the transmit and receive beamforming vectors, which are designed for the intended channel, are independent of the interference channels. This condition is realistic in most practical scenarios.

For the interfering active {UE} in the $k^{th}$ cell, denote the interference channel as ${\mathbf{{H}}_{k}^{}}$ and the unit-norm interfering beamforming vectors as $\textbf{w}_{k} \in \mathbb{C}^{N_{\text{UE}}\times 1}${, which can represent either a baseband digital precoder or analog RF beamforming or a hybrid of both of these forms}. 
\subsubsection{Rich Scattering Environment}
Given that $\textbf{w}_{k}$ depends on the direct channel between the $k^{th}$ active UE and its associated BS, it is independent of the interference channel ${\mathbf{{H}}_{k}^{}}$. We can denote the effective interference vector from the $k^{th}$ interfering active {UE} as
\begin{align}\label{NewChannel}
%
%
{{\mathbf{{h}}}_{k}^{}}&={\mathbf{{H}}_{k}^{}} \textbf{w}_{k} = \sqrt{\mathit{l}\left({\Vert \bold{z}_k \Vert}_2\right)} \mathbf{\tilde{H}}_{k} \textbf{w}_{k} = \sqrt{\mathit{l}\left({\Vert \bold{z}_k \Vert}_2\right)} {{\mathbf{\breve{h}}}_{k}^{}},&&
%
%
\end{align}
where ${{\mathbf{\breve{h}}}_{k}^{}}$ is a random vector with i.i.d. elements as $\mathcal{CN}(0,1)$. 
{This follows from the fact that the inner product of a vector of i.i.d. complex Gaussian variables  (the column of $\mathbf{\tilde{H}}_{k}$) with an arbitrary unit norm vector is a complex Gaussian variable of the same distribution as each element in the original vector.}

The interference vector {\color{black}at the receiving antenna array as} presented in \eqref{Eq_4r1} can be expressed as
\begin{align}
\!\!\!\!\!\!\textbf{v}_{0}&\!=\sum\limits_{k\neq 0} {{\mathbf{h}}_{k}^{}} x_{k} = \sum\limits_{k\neq 0} \sqrt{P_k \mathit{l}\left({\Vert \bold{z}_k \Vert}_2\right)} \mathbf{\breve{h}}_{k}^{} U_k,\label{eq11a_2}
\end{align}
where the summation is over all interfering users.  
\subsubsection{Limited Scattering Environment}
In this environment, we assume directional analog beamforming only and consider the composite effect of both transmit beamforming vector, $\textbf{w}_{k}$, and receive combining vector, $\bar{\textbf{w}}_{0} \in \mathbb{C}^{N_{\text{BS}}\times 1}$, on the received interference power. We further approximate the BSs and UEs antenna array patterns by a sectored antenna model with $M$ and $m$ representing the main and back lobes gains, and $\theta_M$ representing the beamwidth of the main lobe. Then, given the independence between $\textbf{w}_{k}$, $\bar{\textbf{w}}_{0}$, and ${\mathbf{{H}}_{k}^{}}$, the received interference signal after both transmit and receive beamforming is
\begin{align}
\!\!\!\!\!\!\bar{v}_0 = \bar{\textbf{w}}_{0}^{\ast} \textbf{v}_{0}&\!=\sum\limits_{k\neq 0} \bar{\textbf{w}}_{0}^{\ast} {\mathbf{{H}}_{k}^{}} \textbf{w}_{k} x_{k} = \sum\limits_{k\neq 0} \sqrt{G_k P_k \mathit{l}\left({\Vert \bold{z}_k \Vert}_2\right)} a_k  U_k,\label{Eq::Limited_Scatter_Int}
\end{align}
where $a_k$ is the effective channel coefficient and $G_k$ is the effective antenna gain from the $k^{th}$ interfering active UE and can be modeled, assuming channel clusters are well separated, as \newpage
\begin{align}
G_k = \begin{cases} 
MM, & \mbox{with probability } p_1 = K(\frac{\theta_M}{2\pi})^2 \\ 
mM, & \mbox{with probability } p_2 = 2K(\frac{2\pi-K\theta_M}{2\pi})(\frac{\theta_M}{2\pi}) + (K^2-K)(\frac{\theta_M}{2\pi})^2 \\
mm, & \mbox{with probability } p_3 = (\frac{2\pi-K\theta_M}{2\pi})^2
\end{cases}
\end{align}
where $K$ is the number of clusters (the special case of a single cluster channel model, $K=1$, is discussed in \cite{BaiH14,turgut2016coverage}). Note the difference between the formulations in \eqref{eq11a_2} and \eqref{Eq::Limited_Scatter_Int} in that \eqref{eq11a_2} is an interference vector at the receiving antenna array, before receiver processing, whereas the interference in \eqref{Eq::Limited_Scatter_Int} applies after receive beamforming. 
\subsubsection{Modeling the Interference}
For each spatial and channel fading realization, 
based on the large number of interferers, the out-of-cell interference can be modeled as a complex Gaussian random vector for \eqref{eq11a_2} or a complex Gaussian scalar for \eqref{Eq::Limited_Scatter_Int} with zero mean and covariance $\boldsymbol{\Sigma}_{0}$ or variance ${\sigma}_{0}^2$. These covariance and variance, however, {are random and dependent on user locations and channel fading}. The matrix $\boldsymbol{\Sigma}_{0}$ is symmetric with diagonal elements representing the interference power at each receiving antenna element, and the off-diagonal elements representing the correlation among interference signals at different antenna elements. {\color{black} The variance ${\sigma}_{0}^2$ on the other hand represents the total interference power after receiver beamforming.}
We analyze the off-diagonal and diagonal elements of the covariance matrix separately in Sec. \ref{AnalyticMoments}. We show later in the numerical analysis section that the correlation among antenna elements is weak and negligible; hence we focus on modeling the interference power elements in $\boldsymbol{\Sigma}_{0}$ and ${\sigma}_{0}^2$.
\vspace{-.2in}
\subsection{Achievable Rate and Outage Probability Formulation}\label{DomBeam}\vspace{-.1in}
{\color{black}The main goal of this paper is to model the interference as formulated above. In order to evaluate this model, we apply it in capacity analysis. In this section, we formulate the capacity and show how it is affected by interference.}


{\color{black}To establish the per-user capacity in the rich scattering environment, we assume the transmit beamforming vector $\mathbf{w}_{0}$ as the right singular vector corresponding to the dominant mode of $\textbf{H}_{0}$, and consider either an interference-aware (IA) or interference-unaware (IU) receive combining vectors.} 
The noise plus interference in the received signal in \eqref{Eq_4r1} can be treated as a Gaussian random vector with a random covariance matrix $\textbf{R}_0$, which is dependent on interfering nodes locations and their interference channels. This covariance matrix can be expressed as \vspace{-.1in}
\begin{eqnarray}
\textbf{R}_{0}=\boldsymbol{\Sigma}_{0}+\sigma^2 \mathbf{I},
\end{eqnarray}
where $\boldsymbol{\Sigma}_{0}$ is the interference covariance matrix discussed in Sec. \ref{AnalyticMoments}. 
{\color{black} 
Applying a receive combining vector $\textbf{s}_d \in \mathbb{C}^{N_{\text{BS}}\times 1}$ to the received signal ${\textbf{y}}_0$ in \eqref{Eq_4r1} to get 
\begin{eqnarray*}
\tilde{\textbf{y}}_0 = \textbf{s}_d^{\ast} \textbf{y}_{0} = \textbf{s}_d^{\ast} \textbf{H}_{0} \textbf{w}_{0} \sqrt{P} U_0 + \textbf{s}_d^{\ast} \textbf{R}_{0} \tilde{\textbf{z}}_{0},
\end{eqnarray*}
where $\tilde{\textbf{z}}_{0} $ is a vector of i.i.d. zero mean and unit variance complex normal entries, we can then write the per-user capacity for the given combining vector $\textbf{s}_d$ as
\begin{eqnarray}\label{EQ::Original_Rate}
C&\!\!\!\!=\!\!\!\!& \log{\left( 1+ \frac{|\textbf{s}_d^{\ast} \mathbf{H}_{0} \mathbf{w}_{0}|^2}{\textbf{s}_d^{\ast} \mathbf{R}_0 \textbf{s}_d} P_{} \right)}.
\end{eqnarray}

For interference-unaware (IU) combining, we design the receive combining vector without knowledge of the interference statistics -- particularly its covariance. As such, the IU combining vector is chosen  as the left singular vector corresponding to the dominant mode of $\textbf{H}_{0}$, $\textbf{s}_d = \mathbf{v}_1(\textbf{H}_{0})$, resulting in the following per-user capacity:
\begin{align}\label{Eq::RateDEq2}
\!\!\!\!\!\!C_{IU}=  {\log_2}{\left( 1+\frac{|\lambda_{max}(\textbf{H}_{0})|^2}{\sigma^2 + \mathbf{v}_1^\ast \boldsymbol{\Sigma}_{0} \mathbf{v}_1} P_{} \right)},
\end{align}
where $\lambda_{max}(\textbf{H}_{0})$ is the maximum singular value of $\textbf{H}_{0}$. 

As seen in capacity formula \eqref{Eq::RateDEq2}, with IU combining, the effective post-combining interference is $\sigma_I^2 = \mathbf{v}_1^\ast \boldsymbol{\Sigma}_{0} \mathbf{v}_1$. Since $\boldsymbol{\Sigma}_{0}$ is a random covariance, this post-combining interference is also random. Its distribution, however, can be characterized in the same way as the pre-combining interference power at each antenna element, i.e., the diagonal elements of $\boldsymbol{\Sigma}_{0}$, as stated in Lemma 1. As a result, the IU capacity depends only on the interference power at each antenna element but not on the correlation between interference at different antenna elements.
\begin{MomentsLemma}[Interference Power Distribution for IU Combining]\label{Lemma:PostCombInt}
The post-combining interference power $\!\mathbf{v}_1^\ast \boldsymbol{\Sigma}_{0} \mathbf{v}_1$ is a random variable that has the same distribution as a diagonal element of $\boldsymbol{\Sigma}_{0}$. \vspace{-.4in}
\end{MomentsLemma}
\begin{proof}
See Appendix A\ref{APPx::1} for detail.
\end{proof}

For interference-aware (IA) combining, the receive combining vector can be designed as a function of the interference statistics. 
Define $\tilde{\textbf{s}}_d = \mathbf{R}_0^{1/2} \textbf{s}_d$; then replacing $\textbf{s}_d = \mathbf{R}_0^{-1/2} \tilde{\textbf{s}}_d$ into \eqref{EQ::Original_Rate}, it is straightforward to see that the vector which maximizes \eqref{EQ::Original_Rate} is $\tilde{\textbf{s}}_d^\star = {\left(\mathbf{R}_{0}\right)^{-1/2} \textbf{H}_{0}^{} \textbf{w}_{0}}$. 
%
The interference-aware per-user capacity can then be obtained as
\begin{align}\label{RateDEq}
\!\!\!\!\!\!C_{IA}=  {\log_2}{\left( 1+\norm{\tilde{\textbf{s}}_d^\star}^2 P_{} \right)}.
\end{align}
This IA capacity depends on the complete interference covariance $\boldsymbol{\Sigma}_{0}$; it is therefore dependent on both the interference power at each antenna element and the correlation between interference at different antenna elements.
}

In the limited scattering environment, we can express the per-user capacity, assuming perfect beam alignment between the UE and its serving BS \cite{BaiH14}, as
\begin{align}\label{Eq::RateDEq}
\!\!\!\!\!\!C=  {\color{black}\log_2}{\left( 1 + \frac{M^2  |a_0|^2 \mathit{l}\left({\Vert \bold{p}_0 \Vert}_2\right) P}{\sigma^2 + \sigma_0^2} \right)},
\end{align}
{where $M$ is the antenna array main lobe gain; $a_0$ is the effective complex channel coefficient; $\bold{p}_0$ is the location of the active UE with respect to the BS; and $\sigma_0^2$ is the total interference power. Achieving the capacity in \eqref{Eq::RateDEq} does not require knowledge of the interference covariance matrix.}

The rates in Eqs. \eqref{Eq::RateDEq2} -- \eqref{Eq::RateDEq} are random quantities due to channel fading and interference. 
We investigate the outage probability at a given target rate as a performance metric, defined as 
\begin{eqnarray}\label{OutageEquation}
p_o(R_T)= \mathbb{P}\{C  < R_T\},
\end{eqnarray}
where $R_T$ is the {\color{black}target} rate at which the active UE under consideration transmits.
\vspace{-0.1in}
\section{Interference Modeling Approaches and Fitness Metric}\label{secAnalyse}
In wireless cellular networks at system level simulations, the most difficult part is to simulate interference from other cells to the cell under study. This is mainly due to the processing and memory power required. These simulations can be substantially simplified if we can characterize the interference by a simple parametrized distribution that represents a good fit to the actual interference. In this section, we introduce the approaches we use to model the interference power terms on the diagonal of $\boldsymbol{\Sigma}_{0}$ {\color{black}and ${\sigma}_{0}^2$}. Then, we introduce the information-theoretic based metric used to measure the fitness of each of those models. Finally, we discuss candidate distributions considered for the interference model, the {Inverse Gaussian and Inverse Weibull} distributions.
\vspace{-.1in}
\subsection{Interference Modeling Approaches} \label{RelativeEntropy}
We introduce three different approaches for modeling the distribution of 
the interference power at the receive antenna. These approaches are {Moment Matching (MM)}, {\color{black}individual} distribution {MLE} (individual MLE), and mixture distribution MLE (mixture MLE). 

\subsubsection{MM Approach}
In this approach, we leverage tools from stochastic geometry to analytically derive the first two moments of each interference power term. Then, we match these two moments with those of the two-parameter candidate distributions to estimate these two parameters and hence specify the distribution. This approach is the simplest and does not require complex computations to estimate the distributions parameters.
{Thus, it is often used in the literature to produce analytically tractable models \cite{MM_6831093, heathHetero6515339,MM_6410048}. We will show, however, that this approach may not produce a good fit for the model.}
\subsubsection{Individual MLE Approach}
In this approach, to simplify parameter estimation but improve the model fitness, we use a combination of MM and MLE{, a powerful estimation technique which is often used in signal detection and estimation \cite{MLE_signal_poor2013introduction}}. Here we only match the first moment using the analysis results derived based on stochastic geometry, which helps in either determining one of the parameters for the candidate distribution, or getting one parameter in terms of the other. 
Then, we use MLE to determine the other parameter of the candidate distribution.
\subsubsection{Mixture Approach}
This approach is the most complex but also most accurate, in which we model the interference probability density function as a mixture of the two candidate distributions. We again use a combination of MM and MLE techniques to simplify parameter estimation. In this approach, we do not determine the optimal values of the distributions parameters in closed form, but design an efficient {Expectation Maximization (EM)} iterative algorithm to identify these paramters, as discussed in Sec. \ref{EMalgorithmEstimation}.
\vspace{-.1in}
\subsection{KL Divergence Fitness Metric}
In order to measure the {developed} models goodness or fitness, we use a concept from information theory as the relative entropy {\cite{CoverRelativeEntropy}}. {\color{black}In this paper, we measure the relative entropy between a proposed model and the data  obtained by simulating the network based on stochastic geometry as discussed in Sec. \ref{GeoModel}, but it should be noted that other data such as those obtained from measurement campaigns can also be used to test our models via relative entropy.} The relative entropy of a distribution $P$ with respect to another distribution $Q$, also called Kullback-Leibler divergence $D_{KL}(P || Q)$, reflects the difference, or distance, between these two probability distributions. The distribution $P$ usually represents the true distribution of data or observations, while $Q$ is typically used to represent a modeled distribution or an approximation to the true distribution $P$. This relative entropy is defined for continuous probability distributions as
\begin{eqnarray}
D_{KL}(P || Q)= \int_{-\infty}^{\infty} p(x) \log \frac{p(x)}{q(x)} dx
\end{eqnarray}
where $p(x)$ and $q(x)$ denote the probability density functions of $P$ and $Q$, respectively. This definition can also be described as the expectation of the logarithmic difference between the probabilities $P$ and $Q$. 
{Kullback-Leibler divergence metric has been used to measure the accuracy of modeling composite fading and shadowing wireless channels such as the Rayleigh-Inverse Gaussian and mixture Gamma distributions \cite{KL_IG_agrawal2007efficacy,KL_6059452}.}
\vspace{-.1in}
\subsection{Candidate Distributions for the Interference Model}
We consider two candidate distributions for interference modeling: the {inverse Gaussian} (IG) and the {inverse Weibull} (IW), both are characterized by two parameters. 
The IG distribution is also known as the Wald distribution and is usually used to model nonnegative positively skewed data. {For example, a composite Rayleigh-IG distribution is used to approximate the composite Rayleigh-Lognormal distribution in \cite{KL_IG_agrawal2007efficacy}}.
The name "Inverse Gaussian" comes from the inverse relationship
between the cumulant generating functions of these distributions and those of Gaussian distributions. It is a two-parameter family of continuous probability distributions which is specified by a shape parameter ${{\lambda_{}}}$ and a scale parameter ${{\mu_{}}}$. Given an IG random variable $\gamma_{IG}({{\mu_{}}},{{\lambda_{}}})$, its probability density function is defined as {\cite{IG1tweedie1957statistical,IG2tweedie1957statistical,IG1}}
\begin{eqnarray}
f_{\gamma_{IG}}(t|{{\mu_{}}},{{\lambda_{}}})=\sqrt{\frac{{{\lambda_{}}}}{2\pi t^3}} \exp \left\{\frac{-{{\lambda_{}}} (t-{{\mu_{}}})^2}{2 {{\mu_{}}}^2 t} \right\}, \quad \quad t > 0, {{\lambda_{}}}, {{\mu_{}}} >0.\label{f_IG}
\end{eqnarray}
The mean and variance of $\gamma_{IG}({{\mu_{}}},{{\lambda_{}}})$ can be written in the following form \cite{IG1} \vspace{-.1in}
\begin{eqnarray}\label{Param1}
\mathbb{E}[\gamma_{IG}]={{\mu_{}}}, \quad \quad
\text{var}[\gamma_{IG}]=\frac{{{\mu_{}}}^3}{{{\lambda_{}}}}.\label{Param2}
\end{eqnarray}

The IW distribution is also known as the complementary Weibull distribution and is usually used to model nonnegative positively skewed data that exhibits a long right tail. It is also specified by two parameters, a shape parameter ${{c_{}}}$ and a scale parameter ${{b_{}}}$ (or $\lambda_{IW}$ in some text). Given an IW random variable $\gamma_{IW}({{b_{}}},{{c_{}}})$, 
its probability density function is defined as \cite{IW1rinne2008weibull,IW2sultan2014bayesian}\vspace{-.1in}
\begin{eqnarray}
f_{\gamma_{IW}}(t|{{b_{}}},{{c_{}}})&\!\!\!=&\!\!\! \frac{{{c_{}}}}{{{b_{}}}} \left( \frac{t}{{{b_{}}}}\right)^{-{{c_{}}}-1}  \exp \left\{- \left(\frac{t}{{{b_{}}}} \right)^{-{{c_{}}}}\right\}, \quad t \geq 0, {{b_{}}}, {{c_{}}}>0. \label{f_IW1}
\end{eqnarray}

\info{changed comma to full stop after Eq. (20)}The mean and variance of $\gamma_{IW}({{b_{}}},{{c_{}}})$ can be written in the following form \cite{IW1rinne2008weibull} \vspace{-.1in}
\begin{eqnarray}\label{Param1b}
\mathbb{E}[\gamma_{IW}]&=&{{b_{}}} \Gamma(1-\frac{1}{{{c_{}}}}), \quad  \text{if } {{c_{}}}>1, \quad \quad
\text{var}[\gamma_{IW}]={{b_{}}}^2 \Gamma(1-\frac{2}{{{c_{}}}}),\quad  \text{if } {{c_{}}}>2,
\end{eqnarray}
where the Gamma function $\Gamma(t)$ is defined as
$\Gamma(t)=\int_{0}^{\infty} x^{t-1}e^{-x}dx$.

In a recent and independent work, the Gaussian, Gamma, Inverse Gamma, and Inverse Gaussian distributions are considered as models for the interference distribution in a stochastic geometry based spatial network \cite{kountouris2014approximating}. Different point processes are considered in this work, including PPP, Strauss process, and Poisson cluster process. 
The inverse Gaussian distribution is shown to be a reasonable model for the interference created by BSs distributed according to the Strauss point and Poisson cluster processes. Further, both the Gamma and IG were shown to be suitable for the PPP network geometry. This prior work, however, considered a bounded path loss model $l(r)=(1+r^\alpha)^{-1}$ and Rayleigh fading with no shadowing.  

In our work, we find through simulations that the Gaussian, Gamma and Inverse Gamma distributions are all poor models for the interference power when shadowing is present. Shadowing introduces large variation in the interference power and causes a medium to heavy tail, which none of these distributions are capable of capturing. Thus we need to find other distributions that can more closely model heavy tail interference. As such, we examine the Inverse {Gaussian} and also introduce the Inverse Weibull as a good candidate for parameterized heavy tail distributions that can capture shadowing with large standard deviation.
\vspace{-0.1in}
\section{Analytical Moments of Interference and Moment Matching Models}\label{SecIV}
In this section, 
we analytically derive the first two moments of the diagonal and off-diagonal elements of the interference covariance matrix. We then use the derived moments of the interference power at each antenna element in rich scattering, and the post receive beamforming interference power $\sigma_0^2$ in the limited scattering environment, to develop the MM interference models.
\vspace{-.1in}
\subsection{Analytical Moments of Interference}\label{AnalyticMoments}
\vspace{-.1in}
\subsubsection{Interference Power}\label{Diag}
In the rich scattering environment, the diagonal elements  of the interference covariance matrix $\boldsymbol{\Sigma}_{0}$ represent the interference power at each antenna element and they are uncorrelated with the off-diagonal elements. A correlation between any two of these diagonal elements exists and can be determined similarly to the second moment of the off-diagonal elements (Lemma \ref{offDlemma}). However, since this correlation is not strong as shown in the numerical analysis in Sec. \ref{GeoModelCorCoef}, 
we model the diagonal elements of $\boldsymbol{\Sigma}_{0}$ as a random vector of independent elements. Each of these diagonal elements can be generally expressed as
\begin{eqnarray}
\!\!\!\!\!\!\mathit{q}_{0}&\!\!\!=&\!\!\!\sum_{k>0:\bold{z}_k\in \Phi_1} {\mathit{l}\left({\Vert \bold{z}_k \Vert}_2\right)}    {g}_{k}^{} P_{k} .\label{eq25}
\end{eqnarray}
where $ {g}_{k}^{}$ are all $ \text{i.i.d. } \mathcal{\exp} (1) $ and represent the power gain of the small-scale Rayleigh channel fading  from the $k^{th}$ interferer.

Similarly, in the limited scattering environment, the interference power can be expressed as
\begin{align}
\!\!\!\!\!\! \sigma_0^2 = \bar{q}_0 = \sum\limits_{k>0:\bold{z}_k\in \Phi_2} {G_k \mathit{l}\left({\Vert \bold{z}_k \Vert}_2\right)} g_k P_k = \bar{q}_0^{(1)} + \bar{q}_0^{(2)} + \bar{q}_0^{(3)},\label{Eq::Limited_Scatter_Int_Power}
\end{align}
where we assume that the interferers are distributed according to PPP $\Phi_2$ with intensity $\lambda_2$. Using the thinning property of random point processes \cite{turgut2016coverage}, from $\Phi_2$, we define three independent PPPs $\mathcal{N}_i, i \in \{1,2,3\}$ with intensities $\tau_i = p_i \lambda_2$. Then, we can define the normalized interference power terms as
\begin{align}
\frac{\bar{q}_0^{(i)}}{G^{(i)}} = \sum\limits_{k>0:\bold{z}_k\in \mathcal{N}_i} {\mathit{l}\left({\Vert \bold{z}_k \Vert}_2\right)} g_k P_k, \label{Eq::Normalized_Int_Power}
\end{align}
where $G^{(i)} \in \{MM, Mm, mm \}$. 
Comparing Eqs. \eqref{eq25} and \eqref{Eq::Normalized_Int_Power}, we note that the distribution of the normalized interference power for each thinned process in \eqref{Eq::Normalized_Int_Power} can be characterized as that of the interference power in \eqref{eq25} by matching the intensity $\tau_i$ to the intensity of the interferers in the rich scattering environment, $\tilde{\lambda}_1$. Hence, we focus on the characterization of the rich scattering environment interference power in \eqref{eq25}, based on which the three separate interference power terms for the limited scattering environment in \eqref{Eq::Limited_Scatter_Int_Power} can be derived. 

Lemma \ref{Theorem2} characterizes the first two moments of the out-of-cell interference power $\mathit{q}_{0}$ at each receiving antenna. 
\vspace{-.09in}
\begin{MomentsLemma}[Interference Power Moments]\label{Theorem2}
For network-wide deployment of MIMO transmit beamforming, the out-of-cell interference power at each antenna of the interfered receiver has the following first and second moments:
\begin{align}
\mathbb{E}\left[\mathit{q}_{0}\right]&= 
         {\nu}{ \zeta_1 }{}  {R_c^{2-\alpha}}, &
\text{var}\left[\mathit{q}_{0}\right]&= 
2 \tilde{\nu}{\zeta_2}{}  {R_c^{2(1-\alpha)}}.\label{var2}
\end{align}
\begin{align}
\!\!\!\!\!\!\!\!\!\!\!\!\!\!\!\!\!\!\!\!\!\!\!\!\!\!\text{where} &&\!\!\!\!\!\!\!\!\!\! \zeta_1&= P_{k}, & \zeta_2&= 2 P_{k}^2,\nonumber\\
\!\!\!\!\!\!\!\!\!\!\!\!\!\!\!\!\!\!\!\!\!\!\!\!\!\!&&\!\!\!\!\!\!\!\!\!\!{\nu}&=\frac{2\pi \lambda_1 \beta  \mathbb{E}[\mathit{L}_{s}]}{\alpha -2}, & \tilde{\nu}&=\frac{\pi \lambda_1 \beta^2 \mathbb{E}[\mathit{L}_{s}^2]}{2(\alpha -1)}.
\end{align}
\end{MomentsLemma}
\begin{proof}
See Appendix B\ref{APPx2} for details.
\end{proof}

\subsubsection{Interference Correlation among Different Antennas}\label{offDiag}
The first moment of the off-diagonal elements of the covariance matrix $\boldsymbol{\Sigma}_{0}$ represents the direct correlation among interference signals at different receiving antenna elements. The second moment, however, represents \info{represent} the correlation between received interference powers.   
These off-diagonal elements are expressed as in \eqref{Eq::eq25b} for $i\neq j$ in Appendix A. 
Assuming antenna elements are spaced sufficiently far apart such that they experience independent fading, 
Lemma \ref{offDlemma} characterizes their first and second moments.
\vspace{-.1in}
\begin{MomentsLemma}[Interference Correlation Moments]\label{offDlemma}
For network-wide deployment of MIMO transmit beamforming, the correlation among interference at different receiving antenna elements of the BS has the following first and second moments:
\begin{align}
\!\!\!\!\!\!\!\!\!\!\!\!\!\!\!\!\!\!\mathbb{E}\left[\tilde{\mathit{q}}_{0}\right]&=
0,\quad \quad \quad \mathbb{E}\left[(\tilde{\mathit{q}}_{0})^2\right]= \frac{1}{2} \text{var}\left[\mathit{q}_{0}\right].\label{var1}
\end{align}
where $\text{var}\left[\mathit{q}_{0}\right]$ is defined in Lemma \ref{Theorem2}.
\end{MomentsLemma}
\vspace{-.1in}
\begin{proof}
See Appendix B\ref{APPx2} for details.
\end{proof}

The interference correlation coefficient $C_{xy}$ can then be written as
%
%
\begin{eqnarray}
C_{xy} = \frac{\tilde{\mathit{q}}_{0}}{\mathbb{E}\left[{\mathit{q}}_{0}\right]},
\end{eqnarray}
Based on Lemmas \ref{Theorem2} and \ref{offDlemma}, we can derive the variance of this correlation coefficient as 
\begin{align}\label{Cvar}
\text{var} \left[ C_{xy}\right] &= \mathbb{E}\left[C_{xy}^2\right]=\kappa \frac{(\alpha-2)^2 \mathbb{E}\left[\mathit{L}_{s}^2\right]}{(\alpha-1) \left( \mathbb{E}\left[\mathit{L}_{s}\right] \right)^2} = \kappa \frac{(\alpha-2)^2}{(\alpha-1)} e^{\frac{\sigma_{SF}^2}{\zeta^2}}, &
\kappa & = \frac{\zeta_2}{8 \pi \lambda_1 \zeta_1^2 R_c^2}.
\end{align}

This variance increases with increasing path loss exponent, $\alpha$, and  shadow fading standard deviation, $\sigma_{SF}$. 
We perform numerical analysis of this correlation coefficient in Sec. \ref{GeoModelCorCoef} and show that the correlation among interference at different antenna elements is in general weak.
\vspace{-.1in}
\subsection{Moment Matching Interference Models} \label{MMmodel}
Given the first two moments of the interference power in Lemma \ref{Theorem2}, we can use moment matching as the simplest and most straightforward method to estimate the parameters of the candidate distributions. This method, however, may not produce the best fit. Specifically, we use the two analytically derived moments to estimate the $({{\mu_{}}}, {{\lambda_{}}})$ and $({{b_{}}}, {{c_{}}})$ parameters of the IG and IW distributions. 
The fitted distribution can then be used to represent the marginal distribution of the interference power at each receiving antenna element, as stated in Lemmas \ref{Lemma2a} and \ref{Lemma2b} next.
\begin{MomentsLemma}[Estimation of IG Parameters using MM] \label{Lemma2a}
Given the derived moments in Lemma \ref{Theorem2}, the shape and scale parameters ${{\lambda_{}}}$ and ${{\mu_{}}}$ of an IG model for the interference power $\gamma_{IG}[{{\mu_{}}},{{\lambda_{}}}]$ can be estimated as \vspace{-.1in}
\begin{eqnarray}\label{proofeq}
{{\mu_{}}}={ \mathbb{E}[\mathit{q}_{0}]},\quad \quad
{{\lambda_{}}}=\frac{\left( \mathbb{E}[\mathit{q}_{0}] \right)^3}{\text{var}[\mathit{q}_{0}]}.
\end{eqnarray}
\end{MomentsLemma}
\begin{proof}
Follows directly from the definition of $\lambda$ and $\mu$  in \eqref{Param1}.
\end{proof}
%
%
%
%
%
%
\begin{MomentsLemma}[Estimation of IW Parameters using MM] \label{Lemma2b}
Given the derived moments in Lemma \ref{Theorem2}, the shape and scale parameters ${{c_{}}}$ and $b_{}$ of an IW model for the interference power $\gamma_{IW}[{{b_{}}},{{c_{}}}]$ can be estimated, assuming ${{c_{}}}>2$, by solving the following equation:
\begin{align}\label{proofeqb1}
&\text{var}[\mathit{q}_{0}]-\left(\frac{\Gamma(1-2/{{c_{}}})}{\Gamma^2(1-1/{{c_{}}})}-1 \right) \left(\mathbb{E}[\mathit{q}_{0}]\right)^2=0, \quad \quad \quad 
{{b_{}}}=\frac{\mathbb{E}[\mathit{q}_{0}]}{\Gamma(1-1/{{c_{}}})}. 
\end{align} 
If the above equation gives a value of $c < 2$, then set $c = 2.01$ and compute $b$ accordingly.
\end{MomentsLemma}
\vspace{-.1in}
\begin{proof}
Follows directly from the equations for $b$ and $c$ in \eqref{Param1b}.
\end{proof}
The equation in \eqref{proofeqb1} can be solved numerically to obtain ${{c_{}}}$, which is then used 
to obtain ${{b_{}}}$.
\begin{remark}
\normalfont {For the expressions in \eqref{Param1b}, we note that the mean is infinite for ${{c_{}}}<1$ and that the variance is infinite for ${{c_{}}}<2$. This is a consequence of the heavy right tail of the IW distribution. Hence, for the IW moment matching in Lemma \ref{Lemma2b}, we assume that ${{c_{}}}>2$.}
\end{remark}  
%
\begin{remark}
\normalfont {For the path loss exponent value $\alpha=2$, we use numerically evaluated mean and variance to perform moment matching since the expressions in Lemma \ref{Theorem2} can not be evaluated analytically at this value of $\alpha$.} 
\end{remark} 
\vspace{-.1in}
\section{Individual and Mixture MLE Interference Models} \label{SecV}
In this section, we apply maximum likelihood estimation to compute the model parameters. We discuss both the individual MLE and mixture interference models and develop an EM algorithm to estimate the model parameters.
Maximum likelihood is one of the dominant methods of estimation in
statistics.
A typical MLE problem is to estimate the parameter set $\theta$ of a candidate distribution $p(y|\theta)$ from a given data set $y$ and can be expressed as an optimization as follows: 
\begin{equation} \label{MLEoptPro}
 \begin{aligned}
 & \arg \underset{\theta}{\max}
 & & \log p(y|\theta) \quad
 & \text{subject to}
 & & \theta \in \Theta,
 \end{aligned}
\end{equation}
where $\log p(y|\theta)$ is the log likelihood 
of the observed data $y$, and $\Theta$ is a closed convex set.
\vspace{-.15in}
\subsection{Individual MLE Interference Models}\vspace{-.1in}
In individual MLE interference model, we use only a single distribution from the candidates to model the interference power at each antenna. We use the maximum likelihood optimization problem in \eqref{MLEoptPro} to estimate the parameters of each distribution. However, to simplify the estimation, we use moment matching first. In moment matching first, we match the mean, $\mu_Y$, of the observed data set to that of each of the candidate distributions. This matching results in the scale parameters ${{\mu_{}}}$ and ${{b_{}}}$ for the IG and IW distributions as \vspace{-.1in}
\begin{align}\label{M1}
{{\mu_{}}}&=\mu_Y, &
{{b_{}}}&=\frac{\mu_Y}{\Gamma(1-1/{{c_{}}})}. 
\end{align}

Next, we replace these expressions of the scale parameters into the log-likelihood function in \eqref{MLEoptPro} for each candidate distribution. In the case of IG distribution, this substitution results in a single-parameter log-likelihood function as follows 
\begin{eqnarray}
 \log f_{\gamma_{IG}}(y|{{\mu_{}}},{{\lambda_{}}}) 
 &=& \frac{1}{2} \log {{\lambda_{}}} - \frac{1}{2} \log 2\pi y^3 - \frac{{{\lambda_{}}}}{2\mu_{Y}^2} \frac{(y-\mu_{Y})^2}{y}=\log f_{\gamma_{IG}}(y|{{\lambda_{}}}). \label{logIG}
 \end{eqnarray}
Hence, the MLE problem in \eqref{MLEoptPro} becomes very simple and has a closed form solution for ${{\lambda_{}}}$. Similarly, in the case of IW, the substitution results in a single-parameter log-likelihood function 
\begin{eqnarray}
 \log f_{\gamma_{IW}}(y|{{b_{}}},{{c_{}}})
 &=& \log \frac{\mu_{Y}^{{{c_{}}}} {{c_{}}}}{{\Gamma(1-1/{{c_{}}})}^{{{c_{}}}}  y^{({{c_{}}}+1)}} - {\left( \frac{\mu_Y}{\Gamma(1-1/{{c_{}}})} \right)}^{{{c_{}}}} y^{-{{c_{}}}} 
 = \log f_{\gamma_{IW}}(y|{{c_{}}}). \label{logIW}
 \end{eqnarray} 
The derivation and final expressions of the optimal parameters, ${{\lambda_{}}}$ and ${{c_{}}}$, that optimize the fit of each candidate distribution are presented in Sec. \ref{EMalgorithmEstimation}.  
\vspace{-.15in}
\subsection{Mixture MLE Interference Model} \label{MixtureModel}\vspace{-.1in}
In our mixture model, 
the probability density function {\color{black}(pdf)} of the interference model is now a mixture of both IG and IW distributions as defined in the next theorem.\vspace{-.1in}
\begin{Theorem}\label{MixtureTheorem}
Given a data set $Y$, the {\color{black}pdf} of the mixture interference model can be written as  \vspace{-.1in}
\begin{eqnarray}\vspace{-.1in}
f_Y(y|\theta)&=&{w_1 f_{\gamma_{IG}}(y|{{\lambda_{}}})} + {w_2 f_{\gamma_{IW}}(y|{{c_{}}})},
\end{eqnarray}
where $\{w_1, w_2 : \sum_{i=1}^2 w_i=1\}$ are the weight parameters for mixing the IG and IW distributions 
and $f_{\gamma_{IG}}(y|{{\lambda_{}}})$ and $f_{\gamma_{IW}}(y|{{c_{}}})$ are as given in Eqs. \eqref{logIG} and \eqref{logIW} with the scale parameters obtained through moment matching first as in \eqref{M1}. This mixture model is a 3-parameter distribution with the parameter set as $\theta=\{w_1, \lambda, c\}$.
\end{Theorem}

Given the pdf of the mixture model in Theorem \ref{MixtureTheorem}, we note that the individual MLE models are special cases of this mixture model. By setting $w_1=1$ , we have the IG individual MLE model and by setting $w_2=1$, we have the IW individual model. 
Subsequently we focus on the mixture model parameters 
estimation, which can be applied to the individual models as well. 
\vspace{-.15in}
\subsection{MLE Models Parameters Estimation} \label{EMalgorithmEstimation}\vspace{-.1in}
We have three parameters to be optimized in the mixture model which is significantly more than the single parameter in the individual MLE model. This mixture leads to a more complex MLE optimization problem and the EM algorithm \cite{dempster1977maximum} becomes an appealing approach.
Here, we identify the complete data $X$ as the observed data $Y$ plus some hidden data $Z$, i.e. $X=(Y,Z)$, where $Y$ is the observed set of points that we model by a weighted IG and IW distribution and $Z\in \{1,2\}$ is a discrete random variable representing the assignment of each data point to the two candidate distributions.
Then, we can define the pdf of the complete data $X \in \mathbb{R}^+$ as
\begin{eqnarray}
f_X(x|\theta)&=&f_X(Y=y,Z=i|\theta)=w_i \phi_i(y|\theta_i),\label{Pyz}
\end{eqnarray}
where $\phi_1(y|\theta_1) = f_{\gamma_{IG}}(y|{{\lambda_{}}})$ and $\phi_2(y|\theta_2) = f_{\gamma_{IW}}(y|{{c_{}}})$.

Next, we present the algorithm for estimating the parameters of the MLE interference models. We focus on the representative case of the mixture MLE model, which includes the individual MLE models as special cases.
 Given $n$ observations, in order to develop the EM algorithm to estimate the parameters of our mixture model, we use Proposition 1 below \cite{chen2010demystified} which can be {easily} verified using the independence assumption, the fact that the observed data $Y=f(X)$ is a deterministic function of the complete data $X$ for some function $f$, and Bayes' rule. \vspace{-.1in}
 \begin{Proposition}\label{Proposition1} \textnormal{\cite{chen2010demystified}}
 Let the complete data $X$ consist of n i.i.d. samples: $X_1,X_2,...,X_n$, which satisfies $f(X|\theta)=\prod_{i=1}^n f(X_i|\theta)$ for all $\theta \in \Theta$, and let $y_i=f(x_i)$, $i=1,2,...,n$, then \vspace{-.05in}
 \begin{eqnarray}
 Q(\theta|\theta^{(m)})=\sum_{i=1}^n Q_i(\theta|\theta^{(m)})
 \end{eqnarray}
 \end{Proposition} 
 
 We also define a responsibility function $\gamma_{ij}^{(m)}$, which represents our guess at the $m^{th}$ iteration of the probability that the $i^{th}$ sample belongs to the $j^{th}$ distribution component, as
 \begin{eqnarray}\label{ResponsibilityF}
 \gamma_{ij}^{(m)}=p(Z_i=j|Y_i=y_i,\theta^{(m)})=\frac{w_j^{(m)} \phi_j(y_i|\theta_j^{(m)})}{\sum_{l=1}^2 w_l^{(m)} \phi_l(y_i|\theta_l^{(m)})}, \quad i\in\{1,2,...,n\}, j \in \{1,2\}
 \end{eqnarray}
 where $\theta_1$ and $\theta_2$ represent the shape parameters $\lambda$ and $c$, respectively.

An EM algorithm includes two steps: the E-step to calculate the conditional expectation of the log likelihood of the complete data, and the M-step to maximize this conditional expectation function. Lemma \ref{Qfunction} provides the final expression for the EM algorithm $Q$-function in the E-step. 
\vspace{-.3in}
\begin{MomentsLemma}\label{Qfunction}
 Let the complete data $X$ consist of n i.i.d. samples: $X_1,X_2,...,X_n$, which satisfies $f(X|\theta)=\prod_{i=1}^n f(X_i|\theta)$ for all $\theta \in \Theta$, and let $y_i=T(x_i)$, $i=1,2,...,n$, 
 then we have 
 \begin{eqnarray}
 Q({{{\theta}}}|{{{\theta}}}^{(m)})&=&\sum_{j=1}^2 n_j^{(m)} \log w_j+ \frac{1}{2} n_1^{(m)} \log {{\lambda_{}}} - \frac{{{\lambda_{}}}}{2 \mu_Y^2} \sum_{i=1}^n \gamma_{i1}^{(m)} \frac{(y_i-\mu_Y)^2}{y_i} \nonumber\\
 &+&{{c_{}}} n_2^{(m)} \log \mu_Y -{{c_{}}} n_2^{(m)} \log \Gamma(1-{1}/{{{c_{}}}}) + n_2^{(m)} \log {{c_{}}} \nonumber\\
 &-& {{c_{}}} \sum_{i=1}^n \gamma_{i2}^{(m)} \log y_i - \left[ \frac{\mu_Y}{\Gamma(1-1/{{c_{}}})}\right]^{{{c_{}}}} \sum_{i=1}^n \gamma_{i2}^{(m)} y_i^{-{{c_{}}}}, \label{The1eq0}
 \end{eqnarray}
 where $n_j^{(m)}=\sum_{i=1}^n \gamma_{ij}^{(m)}, j \in \{1,2\}$; and ${{{\theta}}}=\{w_1, \lambda, c\}$. 
 \end{MomentsLemma}\vspace{-.1in}
 \begin{proof}
 See Appendix C for details.
 \end{proof}

In the M-step, to update our estimate of the parameter ${{{\theta}}}$, we solve the following optimization problem for the optimal ${{{{\theta}}}}^*$: \vspace{-.1in}
\begin{equation}\label{Opt1}
 \begin{aligned}
 & \arg \underset{{{{\theta}}}}{\max}
 & & Q({{{\theta}}}|{{{\theta}}}^{(m)}) \quad
 & \text{subject to}
 & & \sum_{j=1}^2 w_j=1, w_j \geq 0 \quad, j\in \{1,2\}.\\
 \end{aligned}
\end{equation}
The optimal ${{{{\theta}}}}^*$ is then found as in Theorem \ref{Theorem2b}. \vspace{-.1in}
\begin{Theorem} \label{Theorem2b} 
The optimal new estimate of ${{{\theta}}}^{(m+1)}=\left\{w_1^{(m+1)},{{{\lambda}}}^{(m+1)}, c^{(m+1)}\right\}$ that maximizes the Q-function in \eqref{The1eq0} at the $m^{th}$ iteration are determined as follows:  \vspace{-.08in}
\begin{eqnarray}
w_j^{(m+1)} &=& \frac{n_j^{(m)}}{n} \quad , j \in \{1,2\}, 
\quad \quad \quad \quad
{{{\lambda}}}^{(m+1)} = \frac{n_1^{(m)} \mu_Y^2}{\sum_{i=1}^n \gamma_{i1}^{(m)} \frac{(y_i-\mu_Y)^2}{y_i}}  \label{Muj}
\end{eqnarray}
and ${{c_{}}}^{(m+1)}$ is obtained by solving the following equation numerically: \vspace{-.1in}
{\small
\begin{eqnarray}
0 &=& n_2^{(m)} \left[ \log \frac{\mu_Y}{\Gamma(1-\frac{1}{{{c_{}}}})} + \frac{1}{{{c_{}}}} \left[ 1- \psi \left(1-\frac{1}{{{c_{}}}} \right)\right] \right] - \sum_{i=1}^n \gamma_{i2}^{(m)} \log y_i \nonumber\\
&+& \left[ \frac{\mu_Y}{\Gamma \left( 1- \frac{1}{{{c_{}}}} \right)} \right]^{{{c_{}}}} \sum_{i=1}^n \gamma_{i2}^{(m)} y_i^{-{{c_{}}}} \left[ \log \frac{y_i \Gamma(1-\frac{1}{{{c_{}}}})}{\mu_Y} + \frac{\psi(1-\frac{1}{{{c_{}}}})}{{{c_{}}}}\right] \label{C_IW}
\end{eqnarray}}
\end{Theorem}\vspace{-.05in}
\begin{proof}
See Appendix D for details.
\end{proof}

Algorithm 1 summarizes the EM algorithm for the mixture distribution interference model. 
\begin{table}\label{algorithm1}
\small
    \begin{tabular}{ |p{.5cm}| p{2cm} | p{11.8cm} }
    \hline
   \multicolumn{1}{l|} {1.} &\multicolumn{1}{l}{\textbf{Initialization:}} & 
    Initialize $w_j^{(0)}$, ${{{\lambda}}}^{(0)}$, and ${{{c}}}^{(0)}$, $j\in \{1,2\}$, and compute the initial log-likelihood:
    \begin{eqnarray}
    L^{(0)}&=&\frac{1}{n} \sum_{i=1}^n \log \left( w_1^{(0)} f_{\gamma_{IG}} \left(y_i|{{{\lambda}}}^{(0)} \right) + w_2^{(0)} f_{\gamma_{IW}} \left(y_i|{{{c}}}^{(0)} \right) \right). 
\end{eqnarray} \vspace{-.1in} \\ 
\multicolumn{1}{l|} {2.} &\multicolumn{1}{l}{\textbf{E-step:}} & 
    Compute $\gamma_{ij}^{(m)}$ as given in Eq. \eqref{ResponsibilityF} and $n_j^{(m)}=\sum_{i=1}^n \gamma_{ij}^{(m)}$. \\
\multicolumn{1}{l|} {3.} &\multicolumn{1}{l}{\textbf{M-step:}} & 
    Compute the new estimate for $w_j^{(m+1)}$, ${{{\lambda}}}^{(m+1)}$, and ${{{c}}}^{(m+1)}$, $j\in \{1,2\}$ as given in Eqs. 
     \eqref{Muj}, and \eqref{C_IW}, respectively. \\
\multicolumn{1}{l|} {4.} &\multicolumn{1}{l}{\textbf{Convergence check:}} & 
    Compute the new log-likelihood function 
    \begin{eqnarray}
   \!\!\!\!\!\!\!\!\!\!\!\!\!\!\!\!\!\!\!\!\!\!\!\!\!\!\! L^{(m+1)}&\!\!\!\!=&\!\!\!\frac{1}{n} \sum_{i=1}^n \log \left( w_1^{(m+1)} f_{\gamma_{IG}} \left(y_i|{{{\lambda}}}^{(m+1)} \right) +w_2^{(m+1)} f_{\gamma_{IW}} \left(y_i|{{{c}}}^{(m+1)} \right) \right).
\end{eqnarray} \vspace{-.2in} \\  
\multicolumn{1}{l|} {} &\multicolumn{2}{l} {\textbf{Return to 2} if $|L^{(m+1)}-L^{(m)}|\geq \delta$ for a preset threshold $\delta$; \textbf{Otherwise} end the algorithm.} \\\hline
    \end{tabular}
    \vspace{-0.1in}
    \caption*{\textbf{Algorithm 1.} MLE Models EM Algorithm}\vspace{-.6in}
\end{table}

In the case of IG individual MLE model, we can set $w_1=1$ and the resulting $Q$-function in \eqref{The1eq} represents the log-likelihood function of the IG distribution. Hence, 
problem \eqref{Opt1} results in the MLE of a single parameter, ${{\lambda_{}}}$. The optimal value of ${{\lambda_{}}}$ can be obtained from \eqref{Muj} by setting $n_1^{(m)}=n$ and $\gamma_{i1}^{(m)}=1$. Similarly, the optimal value of ${{c_{}}}$ for the IW individual MLE model can be obtained from the numerical solution of \eqref{C_IW} after setting $w_2 = 1, n_2^{(m)}=n$ and $\gamma_{i2}^{(m)}=1$.
\vspace{-.15in}
\subsection{Functional Fitting of MLE Model Parameters to {\color{black}Channel} Propagation Characteristics} 
By performing the MLE parameters estimation, we observe that these parameters can be fitted directly to rather simple functions of the {\color{black}channel} features, including the shadowing variance and path loss exponent. These functions present excellent fit with simulated data for the range of mmWave propagation parameters as measured in recent campaigns \cite{mmWave2,mmWave4,mmWave6515173,mmWave6387266}.

Corollary 1 presents the polynomial function form that we use to fit the MLE parameters to the {\color{black}channel} shadowing variance. Note the same functional form applies to all three parameters of the mixture MLE model, each parameter with a different set of fitting coefficients. \vspace{-.1in}
\begin{corollary}\label{ParameterFitLemma}
For a given path loss exponent, the parameters of the mixture MLE algorithm can be estimated directly as a function of the shadowing standard deviation $\sigma_{SF}$ as \vspace{-.08in}
\begin{eqnarray} \label{FitEq}
\log_{10} \left( \theta \right)= a_3 \sigma_{SF}^3 +a_2 \sigma_{SF}^2 + a_1 \sigma_{SF} + a_0
\end{eqnarray}
where $\theta$ represents any of the parameters $\{w_1, {{c_{}}},{{\lambda_{}}} \}$. The function coefficients for each parameter are
obtained via numerical fitting based on the EM estimated parameter values. Parameters ${{b_{}}}$ and ${{\mu_{}}}$ can then be obtained analytically based on Eq. \eqref{M1} and Lemma \ref{Theorem2}.
\end{corollary} \vspace{-.1in}
In Appendix E, we provide representative values of the functional coefficients fitted against simulation data based on stochastic geometry. 

Similarly, the MLE parameters can be fitted to functions of the {\color{black}channel} path loss exponent using a simple least square optimization procedure. Corollary 2 presents the polynomial function forms that we use to fit the MLE parameters to the path loss exponent.
\vspace{-.1in}
\begin{corollary}\label{ParameterFitLemma2}
At a fixed value of the shadowing variance, the IG weight $w_1$ and the IW shape parameter $c$ of the mixture MLE model can be estimated directly as a $4^{th}$ order polynomial function of the path loss exponent $\alpha$. 
The logarithm of the IG shape parameter $\log_{10}(\lambda)$ can be estimated as a polynomial function of $\alpha$ of the $3^{rd}$ degree.
The coefficients for each function can be
obtained via least square fitting based on the EM estimated parameter values.
\end{corollary}\vspace{-.1in}

Using these functional fits, we can bypass the MLE iterative algorithm in Algorithm 1 and directly use the table of fitted function coefficients. These functional fits therefore provide a simple and powerful way of modeling the interference and its direct dependence on the {\color{black}channel} propagation features.

Note that the coefficients examples in Appendix E are verified against simulated data using stochastic geometry based network settings and need further verification against actual measurement data, for which the coefficient values may change slightly. We expect, however, that the identified function forms in \eqref{FitEq} remain a good fit. To fit the function coefficients, apply the MLE algorithm to the measured data to identify the interference model parameters. Then fit these model parameters to the function forms in \eqref{FitEq} to identify the coefficient values. Once the fit is verified, the function  with the identified coefficients can be used directly to model the interference in a new environment with given shadowing variance and path loss exponent, without the need to measure new data or perform the MLE algorithm again.
\vspace{-0.1in}
\section{Numerical Performance Results}\label{results}
\vspace{-0.05in}
In this section, we use simulation data obtained from the stochastic geometry based network setting in Sec. \ref{GeoModel} to numerically verify the validity of the analytical results and proposed interference models. We first verify the validity of our moment matching interference models. We then discuss the fit of the MLE individual and mixture interference models. Last, we apply these interference models to evaluate the outage performance of a cellular system employing MIMO joint transmit-receive dominant mode beamforming. 

{\color{black}Since interference in the limited scattering environment can be derived as the sum of interference components from independent, thinned processes, each of which exhibits similar characteristics as the interference in a rich scattering environment but with a different interferer density as in \cite{turgut2016coverage}, we focus on the rich scattering environment in our simulation.}
For the simulations setting, we consider a typical cell of radius $R_c = 150$ m, 
which is related to the active UEs density $\lambda_1$ as {$\eta \lambda_1 = 0.25 R_c^{-2}$}, where $\eta=1$ in the typical case. Later we vary $\eta$ to examine the effect of user density on performance.
%
%
%
\vspace{-.15in}
\subsection{Interference Correlation among Different Antenna Elements}\label{GeoModelCorCoef}
\vspace{-0.1in}
%
%
%
\begin{figure}[t]
        \centering
        \includegraphics[width=2.5in]{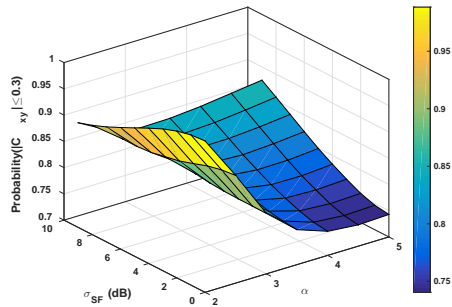}
        \DeclareGraphicsExtensions{.eps}
        \vspace{-0.15in}
        \caption{{\small The probability of weak interference correlation versus shadowing standard deviation and path loss exponent.}}
        \label{corA2_b}
        \vspace{-0.48in}
    \vspace{-.15in}
\end{figure}
This part aims to examine the correlation between interference at different antenna elements to see if this correlation is negligible. We use a kernel distribution to fit the actual interference correlation. A kernel distribution is a nonparametric representation of the probability density function of a random variable. It can be used when a parametric distribution cannot properly describe the simulation data, or when assumptions about the distribution of the data are to be avoided, as is the case for testing the fit of our proposed parameterized interference distributions. 

We examine the interference correlation via the correlation coefficient defined in Sec. \ref{offDiag}. {\color{black} In Fig. \ref{corA2_b}, we plot the probability that the correlation coefficient is smaller than 0.3, and show that this probability is significant for a wide range of the shadowing standard deviation and the path loss exponent.}
%
This figure demonstrates that the interference correlation among different antenna elements is in general weak. {\color{black} As shown later in the system performance evaluation (Sec. \ref{perfeval}), the impact of interference correlation on transmission rate differs depending on the receive combining vector, as discussed in Sec. \ref{DomBeam}. For IA combining, interference correlation can have an impact on capacity; whereas for IU combining, the interference correlation has no impact and can be ignored. In this paper we do not model the interference correlation, hence IU combining is more appropriate.}
%
\vspace{-.15in}
\subsection{Moment Matching Interference Models Evaluation}\vspace{-.1in}
We now use the relative entropy metric described in Sec. \ref{RelativeEntropy} to measure the goodness of the MM interference models in Section \ref{MMmodel}. In Fig. \ref{REAE1a}, we plot the relative entropy versus shadowing standard deviation $\sigma_{SF}$ at a sample path loss exponent value of $\alpha=3.5$. The result shows that IG approximates the simulated interference distribution well at low values of shadowing standard deviation $\sigma_{SF}$ and diverges at higher values. 

The IW distribution, on the other hand, does not show a good fit at the simulated path loss value. Our extensive simulation and testing suggests that IW can be a good approximation of the simulated data distribution only for low values of $\alpha$ (lower than shown in Fig. \ref{REAE1}). 
The moment matching IW distribution diverges significantly as the path loss exponent increases. This divergence is caused by the fact that the optimal shape parameter ${{c_{}}}$, as we will see later in this section, obtained at high values of $\alpha$ is close to or less than $2$. At this range of ${{c_{}}}$ values, the variance of the IW distribution is infinite and the moment matching model is not applicable.

In Fig. \ref{REAE1b}, we show a sample interference power distribution at one antenna element for $\alpha=3.5$ and $\sigma_{SF}=4$ dB to visually confirm the match, or rather the mismatch, of the moment-matched IW interference distribution. The moment-matched IG distribution also starts to diverge from the simulated data at this value of $\sigma_{SF}$, as indicated in Fig. \ref{REAE1a}.
\begin{figure*}[t!]
    \centering
    \begin{subfigure}[t]{0.45\textwidth}
        \centering
        \includegraphics[width=2.5in]{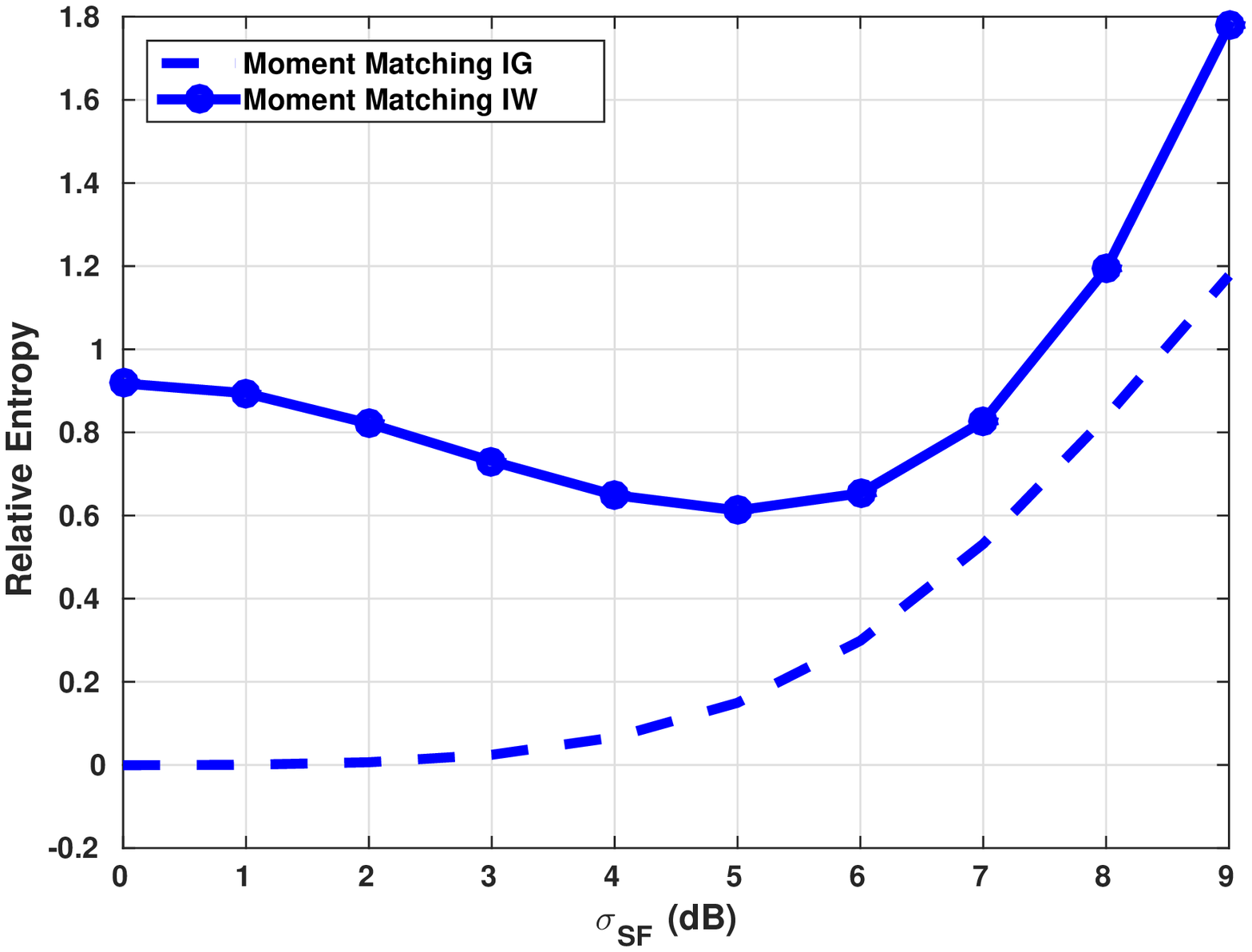}
        \DeclareGraphicsExtensions{.eps}
        \vspace{-0.16in}
        \caption{{\small Relative entropy at $\alpha=3.5$.}}
        \label{REAE1a}
    \end{subfigure}%
    ~ 
    \begin{subfigure}[t]{0.45\textwidth}
        \centering
        \includegraphics[width=2.5in]{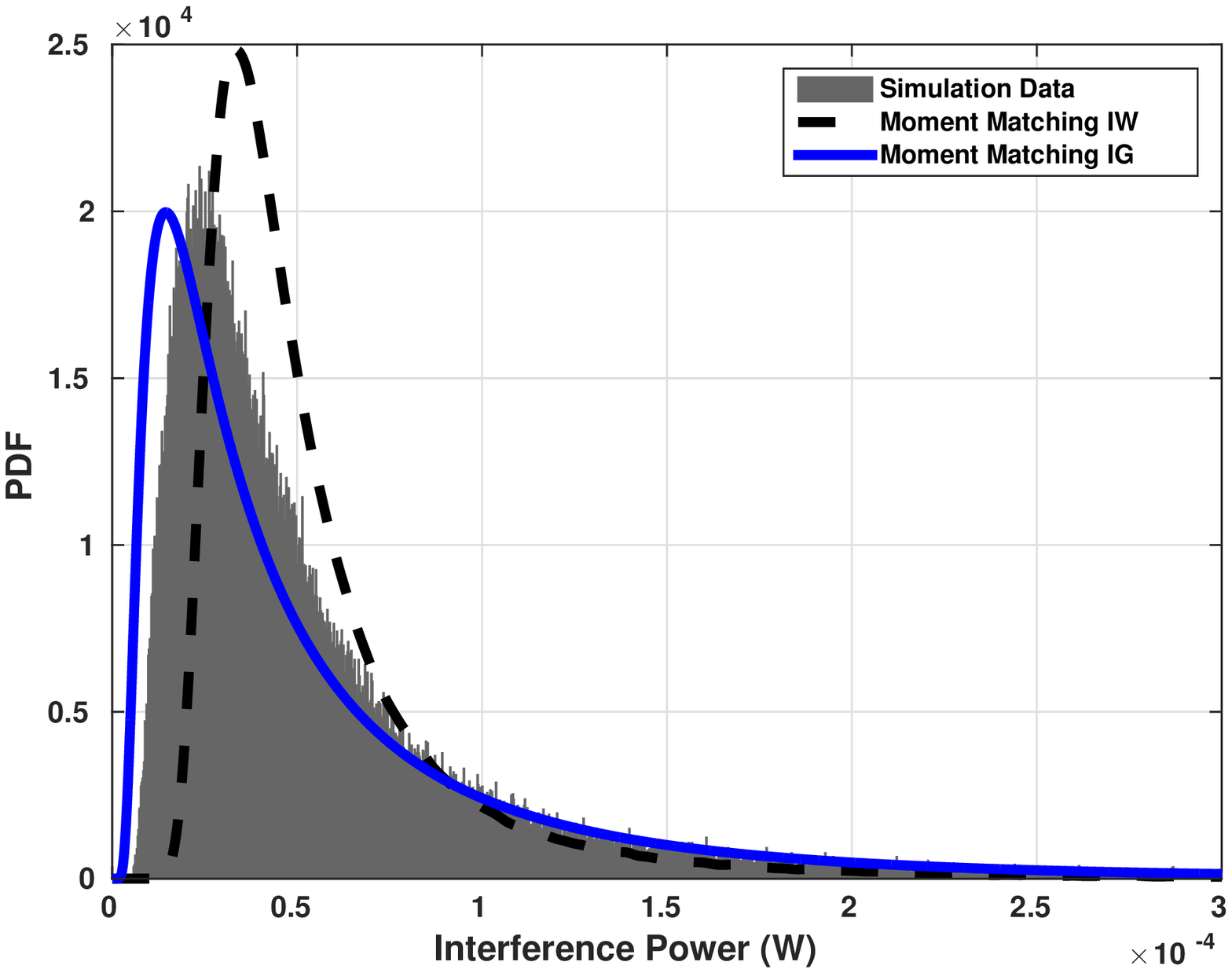}
        \DeclareGraphicsExtensions{.eps}
        \vspace{-0.16in}
        \caption{{\small Sample PDF at $\alpha=3.5$, $\sigma_{SF}=4$ dB.}}
        \label{REAE1b}
    \end{subfigure}\vspace{-0.4in}
    \caption{{\small Relative entropy and sample interference distributions of moment matching interference models. System parameters: $P_{max}=30$ dBm, $\eta=1$, $\delta=10^{-6}$, $\sigma_{SF}=4$ dB, $\alpha=3.5$.}} \label{REAE1}
	\vspace{-0.45in}
	\vspace{-.15in}
\end{figure*}
\vspace{-.15in}
\subsection{Individual and Mixture MLE Interference Models Evaluation}
\vspace{-.1in}
Next, we evaluate the fit of the individual and mixture MLE interference models. Fig. \ref{REAE3} shows the relative entropy measures for these 3 models against simulated data as well as representative sample distributions.
In Fig. \ref{REAE3a}, we see that our proposed mixture model outperforms both the individual IG and IW models. This result aligns with our expectation since the {\color{black}pdf} of the mixture model is a combination of both the individual IG and IW distributions. Hence, performance of the mixture model should be at least as good as that of the individual models. The fact that the mixture model performs better than either individual one is because the relative entropy measure is a non-linear function of the distribution densities.

Comparing between Figs. \ref{REAE1} and \ref{REAE3} also reveals that all MLE models perform much better than the MM models: at the worst $\sigma_{SF}=9$ dB, the IW distribution under individual MLE has a maximum relative entropy of about merely $0.23$ compared to a maximum of $1.8$ under MM. Further, we see that the proposed mixture model follows the simulated data distribution almost exactly even at large shadowing: the maximum relative entropy is barely $0.03$.
\begin{figure*}[t!]
    \centering
    \begin{subfigure}[t]{0.45\textwidth}
        \centering
        \includegraphics[width=2.5in]{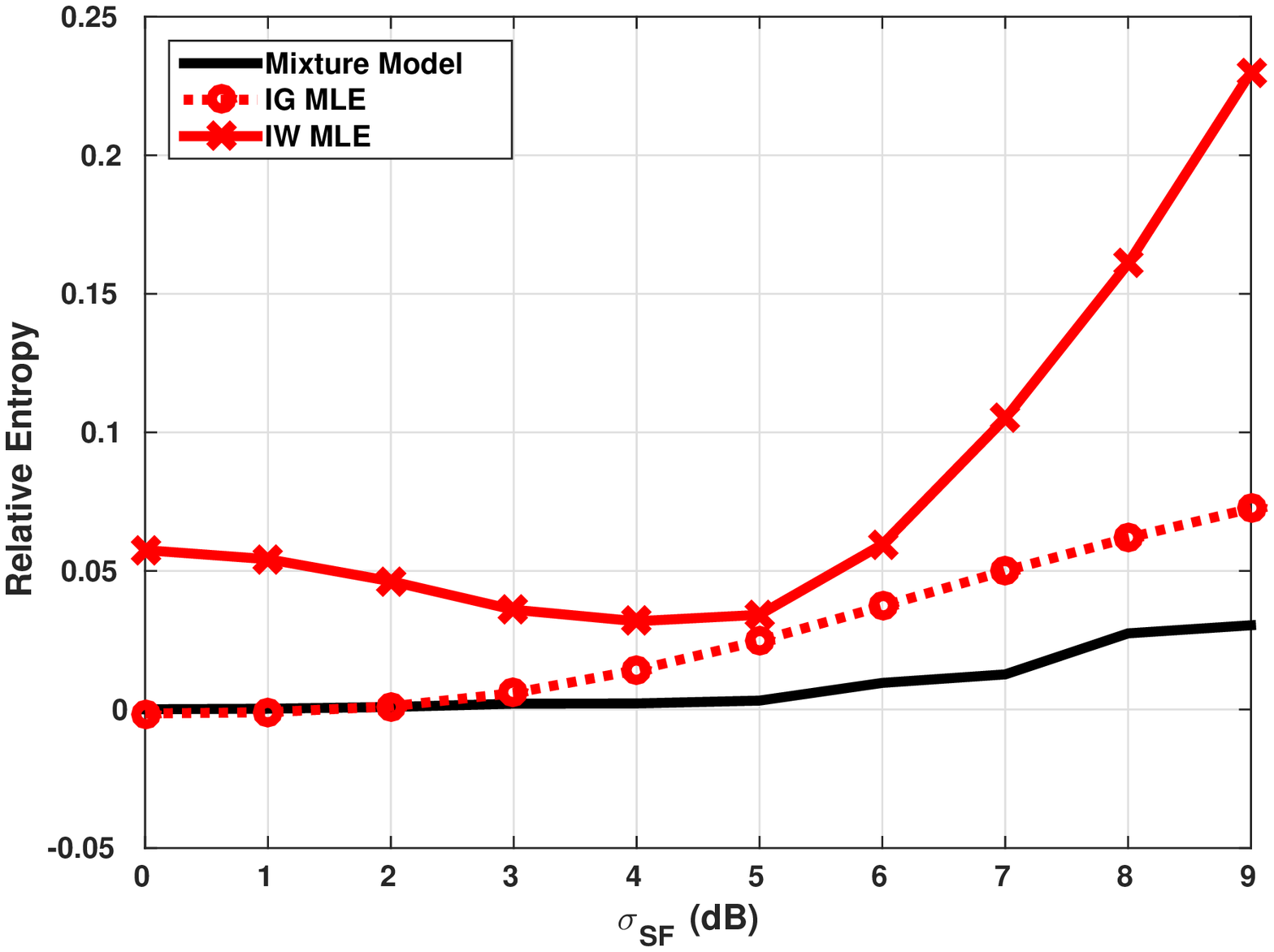}
        \DeclareGraphicsExtensions{.eps}
        \vspace{-0.16in}
        \caption{{\small Path loss exponent $\alpha=3.5$.}}
        \label{REAE3a}
    \end{subfigure}%
    ~ 
    \begin{subfigure}[t]{0.45\textwidth}
        \centering
        \includegraphics[width=2.5in]{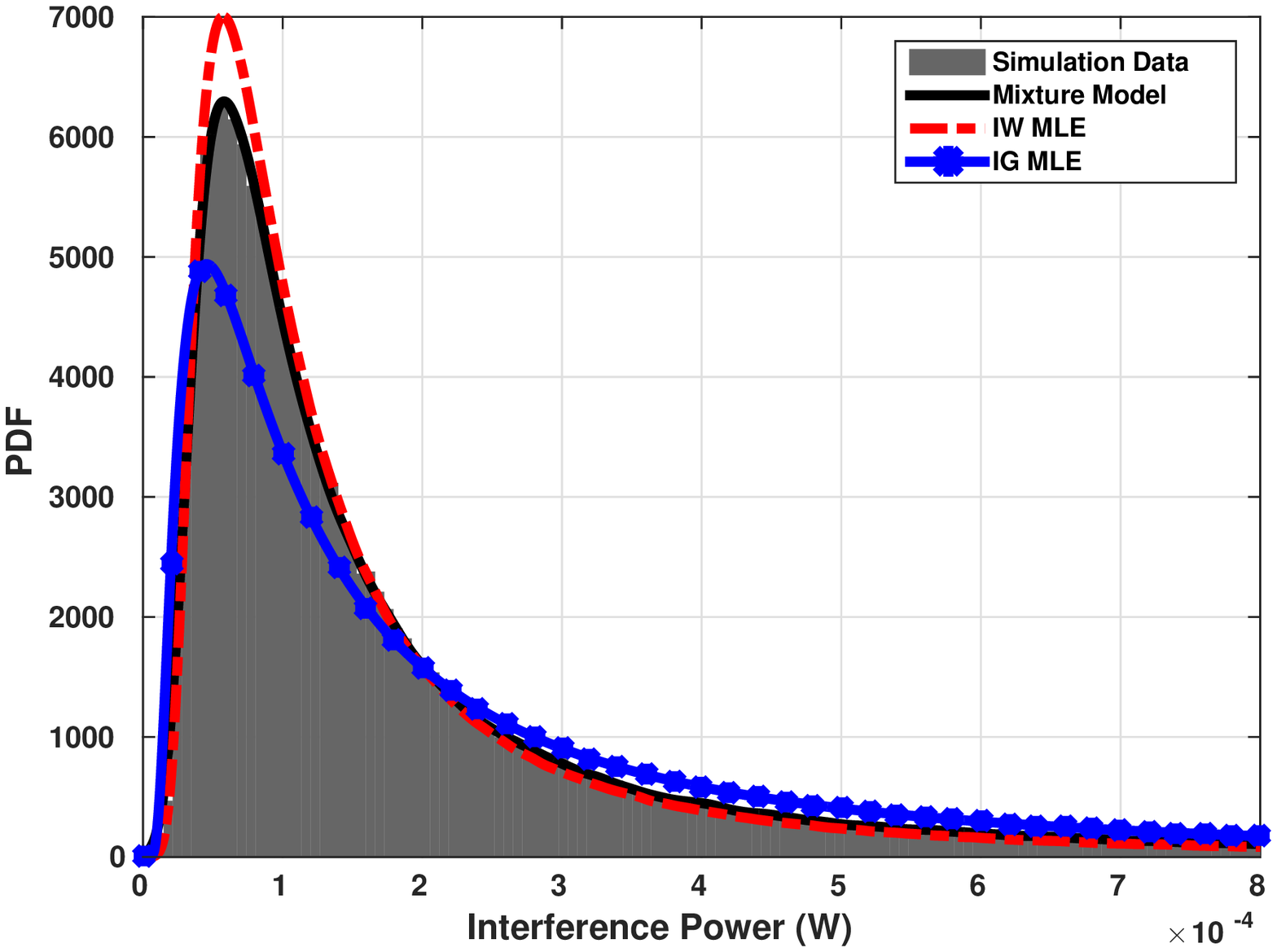}
        \DeclareGraphicsExtensions{.eps}
        \vspace{-0.16in}
        \caption{{\small Sample PDF at $\alpha=3.5$, $\sigma_{SF}=9$ dB.}}
        \label{REAE3b}
    \end{subfigure}\vspace{-0.4in}
    \caption{{\small Relative entropy and sample interference distributions of the individual and mixture MLE models. System parameters: $P_{max}=30$ dBm, $\eta=1$, $\delta=10^{-6}$, $\sigma_{SF}=9$ dB, $\alpha=3.5$.}} \label{REAE3}
	\vspace{-0.45in}
	\vspace{-.15in}
\end{figure*}

Consistently in all our extensive simulations and testing, the mixture MLE interference model provides the best fit. This model requires an iterative algorithm to fit the parameters initially, before fitting the functional coefficients as in Corollary \ref{ParameterFitLemma}. The MM models are attractive from the simplicity point of view since the model parameters can be determined analytically without optimization. 
{\color{black} However, our simulation results suggest that the IG MM model should only be used for the range of shadowing standard deviation $\sigma_{SF}<3$ dB, at which the relative entropy of the IG MM model is approximately within $1\%$ of the mixture model.}
%

{In Fig. \ref{iterations}, we plot the number of iterations required by Algorithm 1 to obtain the optimal values of the mixture MLE model parameters versus $\sigma_{SF}$ at two samples of the path loss exponent $\alpha = \{ 2.5, 4\}$, where the preset threshold value is chosen to be $\delta = 10^{-6}$, the model weights are initialized as $w_1^{(0)}=w_2^{(0)}=0$, and the model shape parameters $\lambda^{(0)}$ and $c^{(0)}$ are initialized to the optimal values of the corresponding individual MLE models. The number of iterations required by the EM algorithm to obtain the optimal mixture MLE model parameters is in general less than 200. In Fig. \ref{Sens12}, we study the sensitivity of Algorithm 1 to the initial values of the parameters. We consider the sample data obtained at path loss exponent $\alpha=2$, shadowing value $\sigma_{SF}=4$ dB, and density factor $\eta = 1$. 
{\color{black}In this figure,} we plot the IG weight parameter $w_1^{(m)}$ at each iteration $(m)$ for different initial values $w_1^{(0)}$, which shows that the weight parameter converges to the same value after at most 170 iterations. We also note from simulations 
that the likelihood function is concave in both of the shape parameters $c$ and $\lambda$ at a fixed value of the weight parameter, therefore it has a unique maximum. Hence, we can conclude based on numerical evidence that Algorithm 1 converges to the same optimal parameter values irrespective of their initializations.}
\begin{figure*}
\centering
\begin{minipage}[b]{.5\textwidth}
\centering
\includegraphics[width=2.5in]{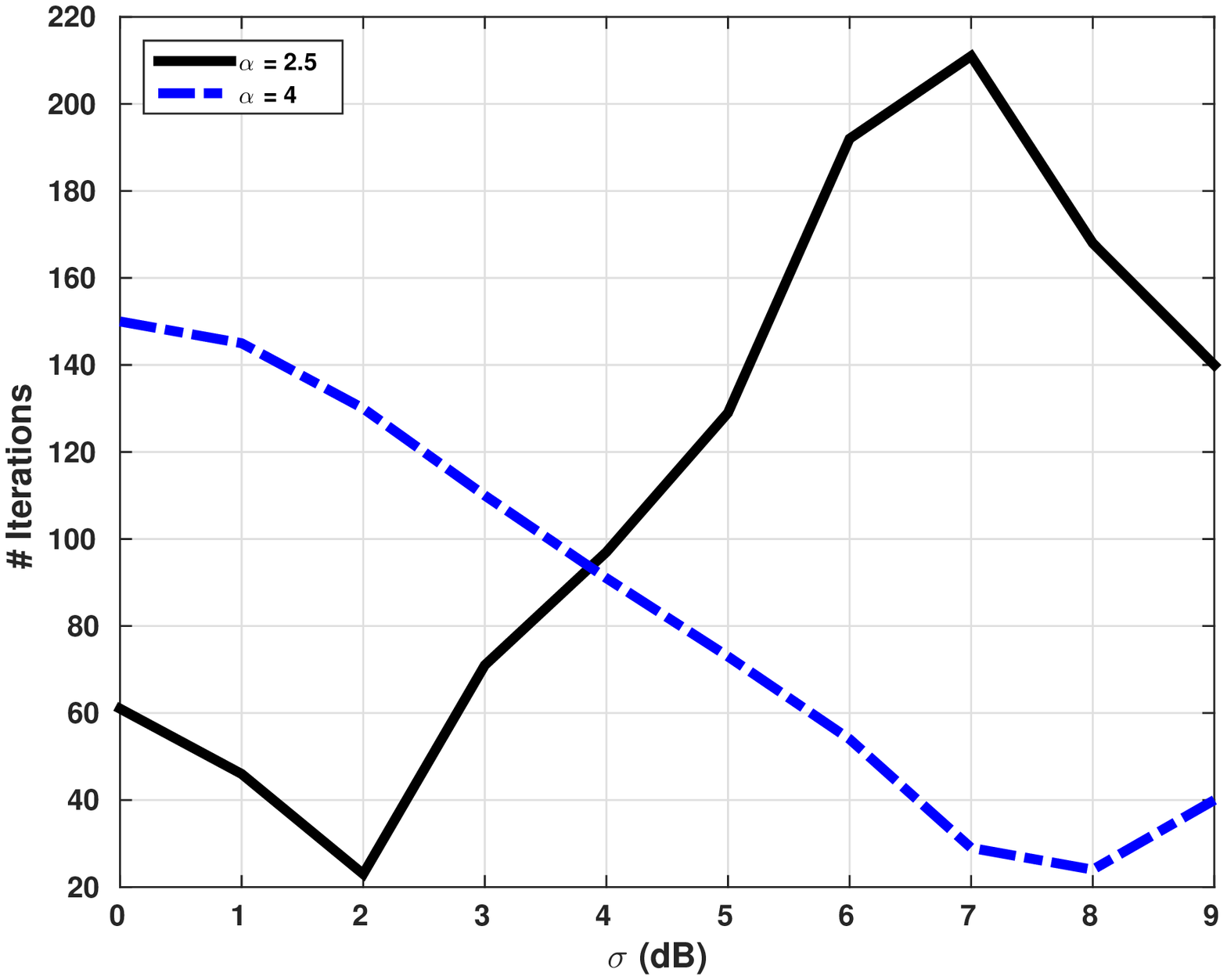}
        \DeclareGraphicsExtensions{.eps}
        \vspace{-0.2in}
        \caption{{\small Number of iteration required to obtain optimal parameters of the mixture MLE model versus $\sigma_{SF}$ at path loss exponent values $\alpha = \{2.5, 4\}$ using the threshold value $\delta = 10^{-6}$.}} 
        \label{iterations}
\end{minipage}\quad
\begin{minipage}[b]{.47\textwidth}
\centering
        \includegraphics[width=2.5in]{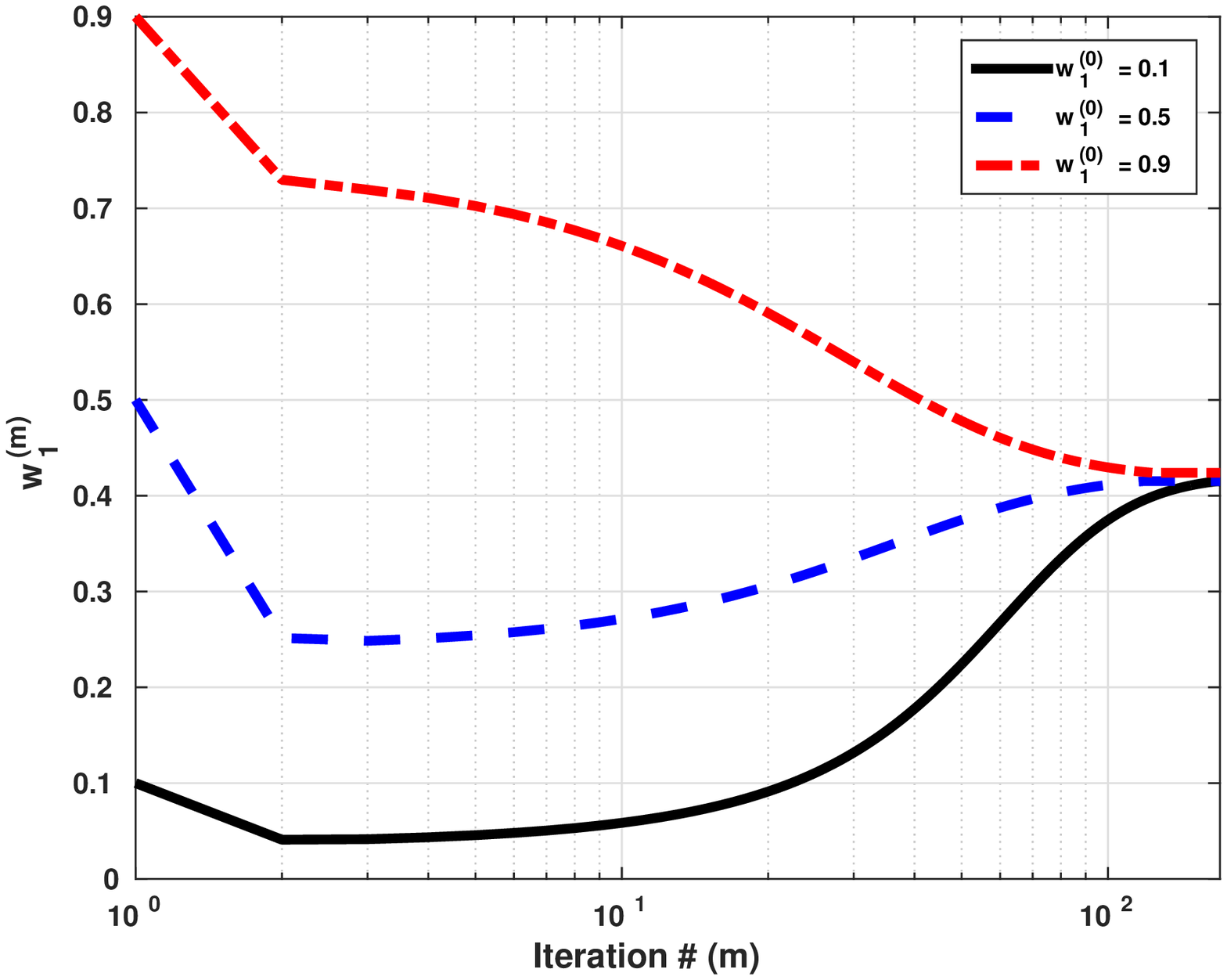}
        \DeclareGraphicsExtensions{.eps}
        \vspace{-0.16in}
        \caption{{\small Sensitivity analysis of the MLE EM algorithm using sample data obtained at path loss exponent and shadowing standard deviation values of $\alpha = 2, \sigma_{SF} =4$ dB and $\eta=1$.}} 
        \label{Sens12}
\end{minipage}\vspace{-0.45in}\vspace{-.15in}
\end{figure*}
\vspace{-.15in}
\subsection{Mixture Model Parameters and Function Coefficients Fitting}
\vspace{-.08in}

In this section, we use the EM algorithm discussed in Sec. \ref{EMalgorithmEstimation}. along with Theorem 2 to obtain the mixture model parameters. We then fit the estimated parameters $w_1, {{c_{}}}, {{\lambda_{}}}, $ into a function in terms of $\sigma_{SF}$ as in 
Corollary \ref{ParameterFitLemma}. After that, Eq. \eqref{M1} can be used to determine the parameter ${{b_{}}}$. We can also use \eqref{var2}
to determine the parameter ${{\mu_{}}}$ directly, but instead we provide a function fitted to the simulated data for this parameter, since \eqref{var2} is invalid at $\alpha=2$. 
The coefficients of the fitted functions are provided for representative values in Appendix E. 

In Fig. \ref{IGwAE13}, we plot the mixture model weight factor obtained from both simulation (cross markers) and functional fit (lines). 
The estimated parameters and fitted functions show an excellent fit to the simulated data. The results confirm that larger shadowing results in an increase in the IW component and a decrease in the IG component in our mixture interference model. The increase in shadowing standard deviation $\sigma_{SF}$ implies a higher probability of high interference power, leading to a heavier interference distribution tail.
As the path loss exponent $\alpha$ increases, however, the IW distribution starts to diverge from the simulated data and hence reduces in weight, thus the IG weight $w_1$ increases as shown. Note also that as $\alpha$ increases, shadowing has less effect on the mixture model weights, especially at low active UEs density. As such, the weight distribution at high path loss is less sensitive to shadowing.  

In Fig. \ref{IWsAE}, we plot the IW and IG shape parameters $\{{c_{}},\lambda\}$ versus $\sigma_{SF}$ at 
$\alpha=3$. 
These curves match the simulated data perfectly, verifying the functional fits in Corollary 1. 
At higher path loss exponents, however, we noted some irregularities in the values of simulated ${{c_{}}}$ 
specifically at $\alpha \geq 4$ and $\sigma_{SF} \leq 4$ 
(not shown in the figure).
This irregularity is caused by the fact that the IG distribution has a much higher weight than IW in this range, hence the IW distribution can assume a looser range of values for its parameters without affecting the goodness of the fit. 
The functional fit for $c$ in \eqref{FitEq} thus can still be applied in this high path loss range.

These plots demonstrate that the fitted functions for parameters of the mixture {MLE model} 
are accurate for the practical ranges of {channel} path loss and shadowing. 
These functions therefore provide an extremely simple yet accurate way to estimate interference distribution parameters.
\begin{figure*}
\centering
\begin{minipage}[b]{.5\textwidth}
\centering
        \includegraphics[width=2.5in]{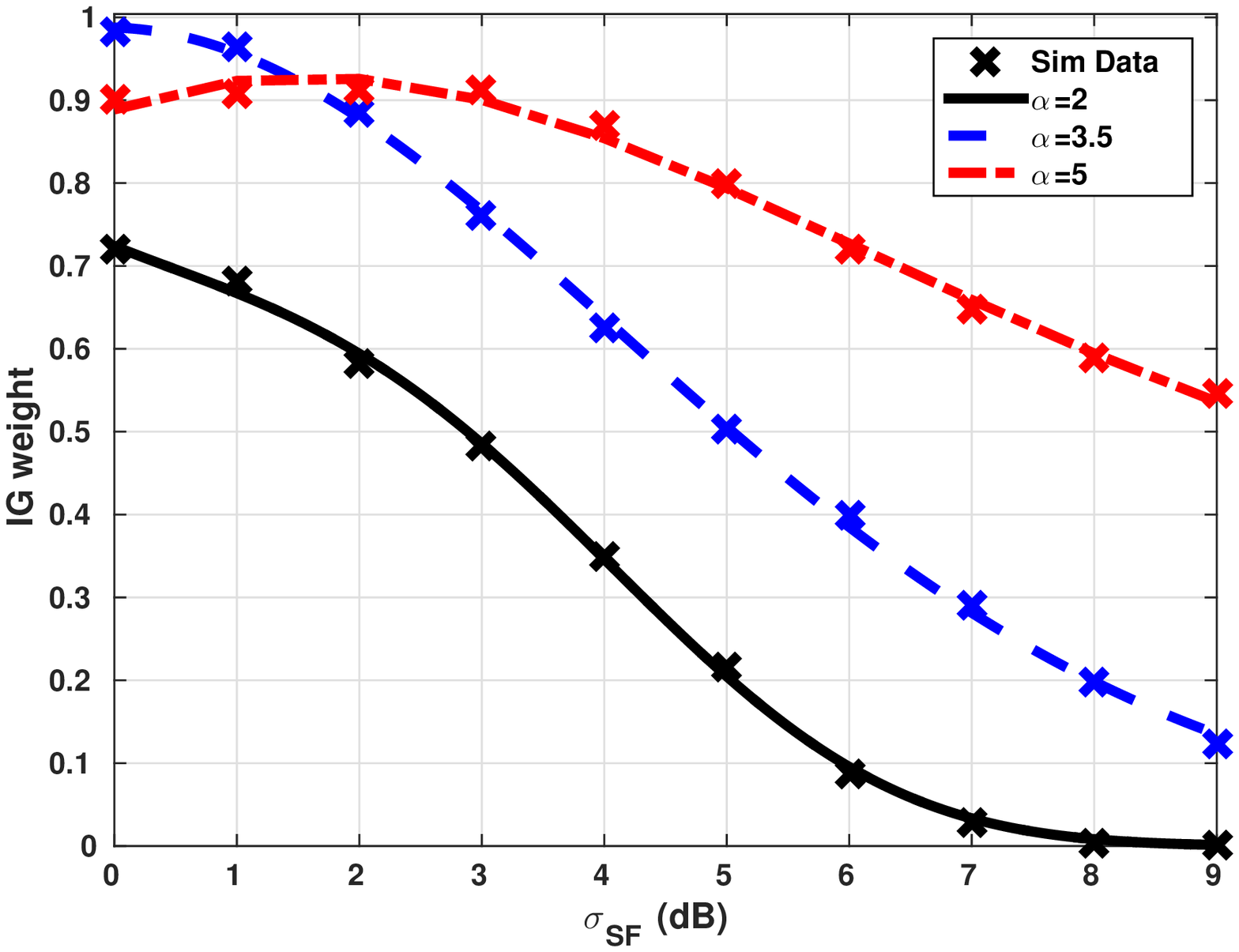}
        \DeclareGraphicsExtensions{.eps}
        \vspace{-0.25in}
        \caption{{\small IG weight, $w_1$, in the mixture MLE model. System parameters: $P_{max}=30$ dBm, $\eta=3$, $\delta=10^{-6}$, $\alpha=\{2, 3.5, 5\}$.\\}} 
        \label{IGwAE13}
\end{minipage}\quad
\begin{minipage}[b]{.47\textwidth}
\centering
        \includegraphics[width=2.5in]{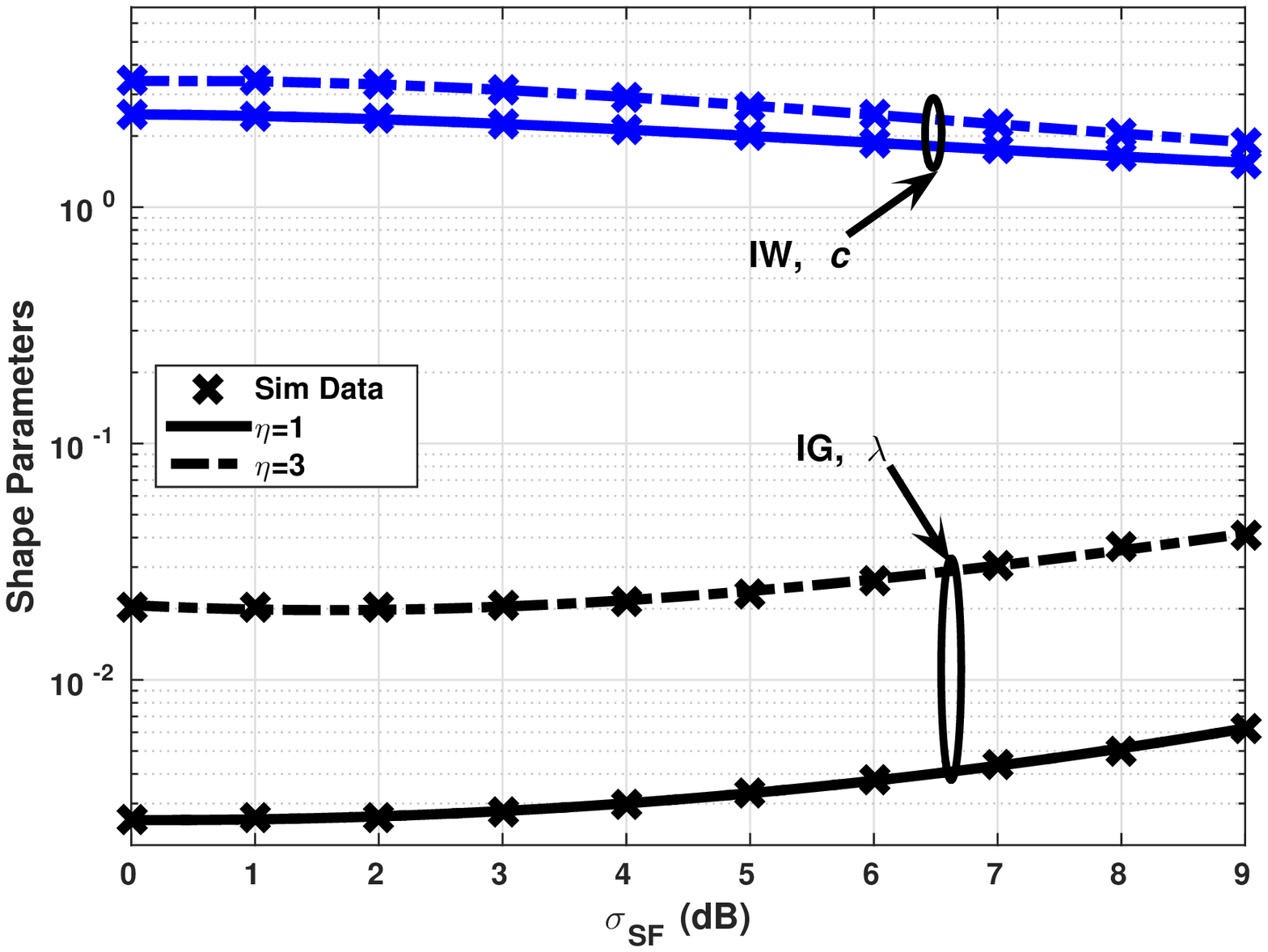}
        \DeclareGraphicsExtensions{.eps}
        \vspace{-0.25in}
        \caption{{\small IW and IG parameters comparison between simulation (markers) and fitted functions in Corollary 1 (lines).  System settings: $P_{max}=30$ dBm, $\alpha=3$, $\eta=\{1, 3\}$, $\delta=10^{-6}$.}} \label{IWsAE}
\end{minipage}\vspace{-0.48in}\vspace{-.15in}\vspace{-.1in}
\end{figure*}
\vspace{-.2in}
\subsection{Rate CDF and Outage Performance}\label{perfeval}
\vspace{-.1in}
We now apply the proposed interference models to evaluate the cellular system performance. In particular, we evaluate the transmission rate cumulative distribution function (CDF), which is related to the outage probability, of a cell edge user at a distance of $145$ m from {the base station} with $\eta=1$ (note $R_c=150$ m and $\eta \lambda_1=0.25 R_c^{-2}$). {In this simulation scenario, we do not consider uplink scheduling; instead, we consider a random user with a random channel to the considered BS. We further normalize the simulation results by the channel bandwidth and show rates in bps/Hz.} 
The uplink maximum transmission power is $30$ dBm. 
We use the dominant mode {\color{black}transmit} beamforming discussed in Sec. \ref{DomBeam} and assume a noise variance of {$\sigma^2=-124$ dBm/Hz} and a path loss intercept $\beta=-72.3$ dB. 

{\color{black} In Fig. \ref{Outs19E1}, we plot the interference-aware capacity in \eqref{RateDEq} and compare the user transmission rate CDFs based on IG MM and mixture interference models to the simulated data. Both models show good accuracy for the interference power as in the left plot at a low $\sigma_{SF}$ value. 
At larger shadowing, however, as in the right plot with $\sigma_{SF}=9$ dB, only the mixture model remains accurate 
while the IG MM model diverges significantly. This result also suggests that interference has a strong impact on user transmission rate CDF even at the somewhat high noise variance of $\sigma^2=-124$ dBm/Hz. We also note that the slight mismatch between the mixture model and simulation is due to the interference correlation, which we did not model but does have an impact on capacity in an interference-aware combining scheme. 

To further examine the accuracy of the mixture interference power model, we consider the interference-unaware combining scheme with capacity in \eqref{Eq::RateDEq2}. IU capacity does not depend on interference correlation as a result of Lemma \ref{Lemma:PostCombInt}. Results in Fig. \ref{Fig::OutageVSelements} shows that the capacity using interference power model matches with simulation almost perfectly in all considered antenna configurations, even at a high shadowing variance of $\sigma_{SF}=9$ dB. This result confirms the validity of our interference power distribution in capacity evaluation.

Our interference distribution models have been validated against simulation based on stochastic geometry. The ideal next step would be to evaluate the models against actual measurement data, but at the point of writing this paper, such data are not available.} As stochastic geometry has been shown to reasonably capture cellular performance compared to actual data \cite{ref_6}, we expect that the mixture model with adjusted functional coefficients will also present a good fit to actual measurement data. The mixture MLE model can provide a valuable tool for cellular system performance evaluation and prediction.

\iftrue
\begin{figure*}[t!]
    \centering
    \begin{subfigure}[t]{0.45\textwidth}
        \centering
        \includegraphics[width=2.5in]{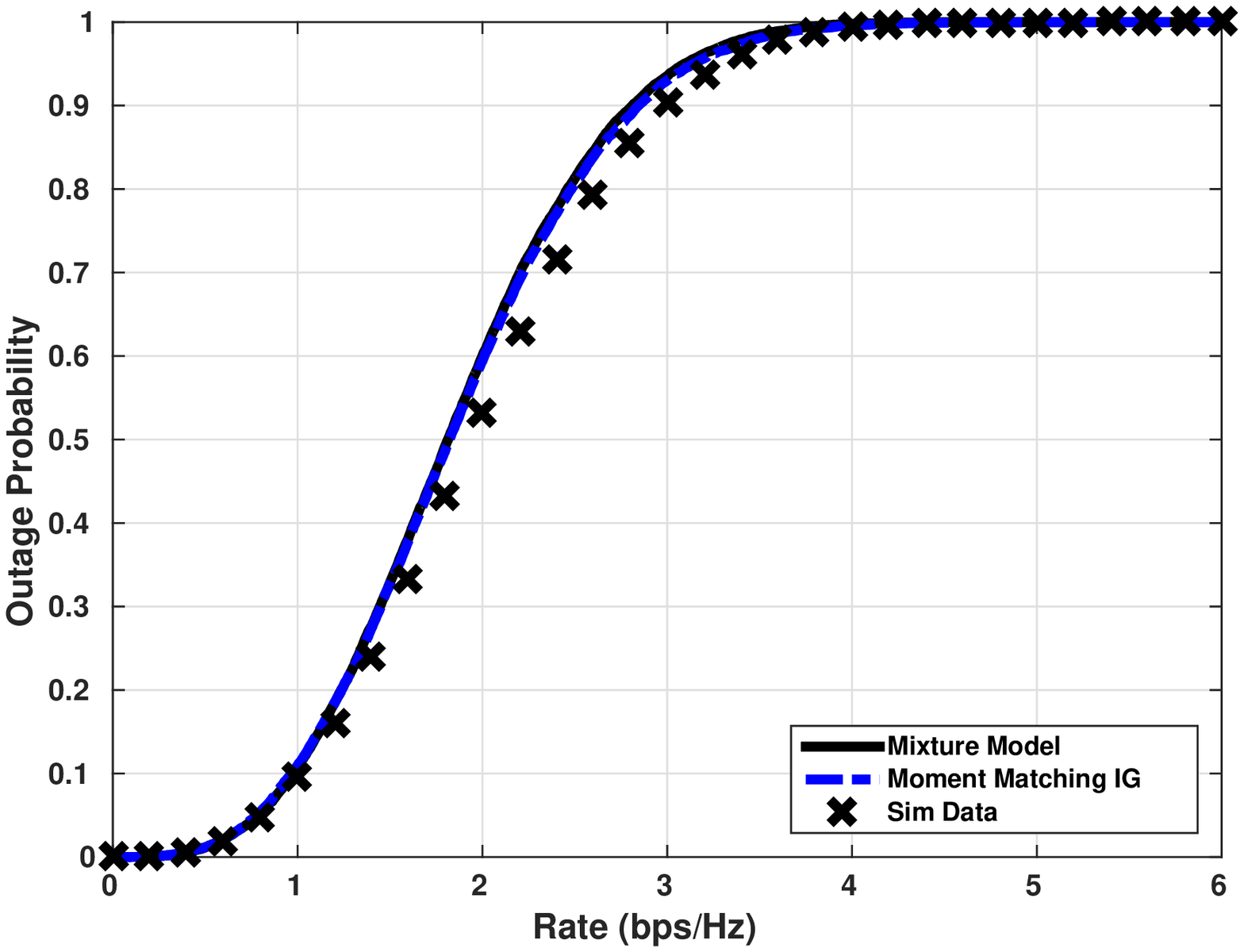}
        \DeclareGraphicsExtensions{.eps}
        \vspace{-0.16in}
        \caption{{\small $\sigma_{SF}=1$ dB.}}
        \label{Outs1E1}
    \end{subfigure}%
    ~ 
    \begin{subfigure}[t]{0.45\textwidth}
        \centering
        \includegraphics[width=2.5in]{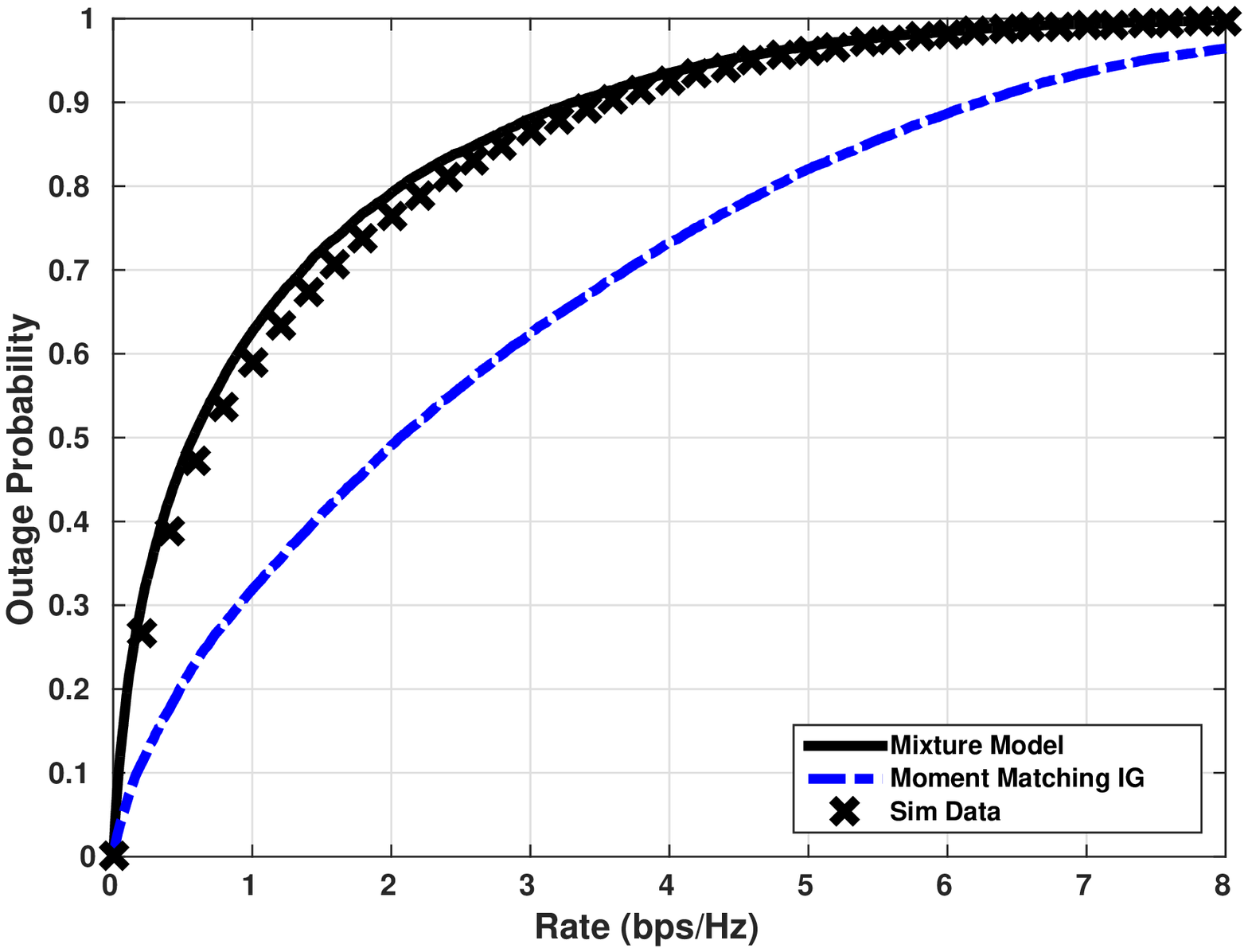}
        \DeclareGraphicsExtensions{.eps}
        \vspace{-0.16in}
        \caption{{\small $\sigma_{SF}=9$ dB.}}
        \label{Outs9E1}
    \end{subfigure}\vspace{-0.4in}\\
    \caption{{\small User transmission rate CDF comparison among simulation, IG MM model, and mixture MLE model. System settings: $P_{max}=30$ dBm, $\eta=1$, $\alpha=3$, $\sigma_{SF}=\{1, 9\}$ dB.}} \label{Outs19E1}
    \vspace{-0.48in}
    \vspace{-.15in}
\end{figure*}
\fi
%
%
\vspace{-.1in}
\section{Conclusion}\label{conclusion}\vspace{-.1in}
We introduced new interference models based on moment matching and maximum likelihood estimation techniques {\color{black}for both the rich scattering and limited scattering NLOS propagations}. 
A novel model as a mixture between the Inverse Gaussian and Inverse Weibull presents remarkably good fit with simulation, capturing the heavy-tail characteristics of interference especially in high shadowing environments. We designed an expectation maximization algorithm to estimate the parameters of this mixture model. 
We also fitted the mixture model parameters to simple polynomial functions of the {channel} propagation features, which provide a simple way to model the interference without any optimization. The fitted mixture model can then be integrated into a more complex system level simulation to evaluate and predict cellular performance, or used to aid system design. 
The next step would be to test the mixture model with appropriately adjusted parameters and functional fitting coefficients against data from measurement campaigns to evaluate the model in an actual system setting. 
\begin{figure}[t]
        \centering
        \includegraphics[width=2.5in]{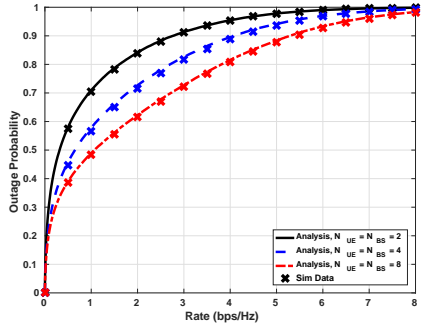}
        \DeclareGraphicsExtensions{.eps}
        \vspace{-0.25in}
        \caption{{\small User transmission rate CDF comparison among simulation and mixture MLE model. System settings: $P_{max}=30$ dBm, $\eta=1$, $\alpha=3$, $\sigma_{SF}=9$ dB, $N_{UE}=N_{BS}=\{2, 4, 8\}$.}} 
        \label{Fig::OutageVSelements}
        \vspace{-0.48in}
    \vspace{-.15in}
\end{figure}
\appendices
{\color{black}
\vspace{-.15in}
\section*{Appendix A: Proof of Lemma \ref{Lemma:PostCombInt}}\label{APPx::1}\vspace{-.05in}
Given the interference vector in \eqref{eq11a_2}, the element in the $i^{th}$ row and $j^{th}$ column of the interference covariance matrix $\boldsymbol{\Sigma}_{0}$ can be expressed as \vspace{-.05in}
\begin{eqnarray}
\!\!\!\!\!\![\boldsymbol{\Sigma}_{0}]_{ij}&\!\!\!=&\!\!\!\sum_{k>0:\bold{z}_k\in \Phi_1} {\mathit{l}\left({\Vert \bold{z}_k \Vert}_2\right)}  \breve{h}_{(k,i)} \breve{h}_{(k,j)}^\ast P_{k},\label{Eq::eq25b}
\end{eqnarray}
where $\breve{h}_{(k,i)}$, $\breve{h}_{(k,j)}$ are i.i.d. $\mathcal{CN}(0,1)$ representing the Rayleigh small-scale fading channels from the $k^{th}$ UE to the $i^{th}$ and $j^{th}$ antenna element of the considered BS, respectively. The post-combining interference power term $\mathbf{v}_1^\ast \boldsymbol{\Sigma}_{0} \mathbf{v}_1$ can then be written as \vspace{-.05in}
\begin{eqnarray}\label{Eq::Post-Comb-Int}
\mathbf{v}_1^\ast \boldsymbol{\Sigma}_{0} \mathbf{v}_1 = \sum_{k>0:\bold{z}_k\in \Phi_1} P_{k} {\mathit{l}\left({\Vert \bold{z}_k \Vert}_2\right)}  \left(\sum_i [\mathbf{v}_1]_i \breve{h}_{(k,i)} \right) \left(\sum_i [\mathbf{v}_1]_i \breve{h}_{(k,i)} \right)^\ast.
\end{eqnarray}
Since the direct channel small-scaling fading component is assumed to be independent of interference channels small-scale fading, $\mathbf{v}_1$ is independent of all $\mathbf{\breve{h}}_{k}$. Thus the distribution of $\sum_i [\mathbf{v}_1]_i \breve{h}_{(k,i)}$ is 
 $\mathcal{CN}(0,1)$ as $\mathbf{v}_1$ is unit-norm, which renders the distribution of $\mathbf{v}_1^\ast \boldsymbol{\Sigma}_{0} \mathbf{v}_1$ in \eqref{Eq::Post-Comb-Int} the same as that of the diagonal elements of $\boldsymbol{\Sigma}_{0}$ in \eqref{Eq::eq25b}.
}
\vspace{-.15in}
\section*{Appendix B: Proof of Lemmas \ref{Theorem2} and \ref{offDlemma}}\label{APPx2}\vspace{-.05in}%
We first derive the $Laplace$ transform of the interference power term, $\mathit{{q}}_{0}$, which is then used to characterize its moments. 
We develop the $Laplace$ transform of the interference power at the destination as follows: 
\begin{align}\label{App1}
\mathcal{L}_{\mathit{{q}}_{0}} (s)&
=\mathbb{E}_{\mathit{{q}}_{0}} \left[e^{-s \mathit{{q}}_{0}} \right] 
%
=\mathbb{E}_{\Phi_1, T} \left[ \prod_{\mathbf{z}_k \in \Phi_1  \setminus \mathbf{z}_0}  { e^{ -s {  {g}_{k}^{} {\mathit{l} \left({\Vert \bold{z}_k \Vert}_2\right)} P_{k}                 }  }  } \right], i \ne j \nonumber\\
\!\!\!\!\!\!\!\!\!\!\!\!\!\!\!\!\!\!\!\!\!\!\!\!&=\mathbb{E}_{\Phi_1} \left[ \prod_{\mathbf{z}_k \in \Phi_1  \setminus \mathbf{z}_0} \mathcal{L}_{\mathcal{J}_{0}} \left(s,{\Vert \bold{z}_k \Vert}_2 \right)  \right]
=\exp \left( {-2\pi\lambda_1 \int_{R_c}^\infty {\left( 1- \mathcal{L}_{\mathcal{J}_{0}}\left(s,r\right) \right) r dr}} \right),
\end{align}
%
where the last equality follows from the $Laplace$ functional expression for PPP using polar coordinates and assuming the field of interferers outside a cell of fixed radius $R_c$; 
$ \mathcal{L}_{T}(s) $ is the Laplace transform of the composite shadowing and small scale fading channel, where $ T \buildrel d \over = \mathit{L}_{s} G$ with $G \sim \exp(1)$ as an exponential random variable, and $ \mathcal{L}_{\mathcal{J}_{0}}\left(s,{\Vert \bold{z}_k \Vert}_2\right) $ is expressed as 
\begin{eqnarray}
\mathcal{L}_{\mathcal{J}_{0}} \left(s,{\Vert \bold{z}_k \Vert}_2 \right)
&\!\!\!\!=&\!\!\!\!    \mathcal{L}_{T}  \left( { s {   {\mathit{l}_{p}\left({\Vert \bold{z}_k \Vert}_2\right)} P_{k} }} \right)  \label{eq29a}
\end{eqnarray}
%
%

Then, we obtain the results in Lemma \ref{Theorem2} based on the interference power $Laplace$ transform and evaluating the following formulas:
\begin{eqnarray}\label{Mformula}
\!\!\!\!\!\!\!\!\!\!\!\!\! \mathbb{E}\left[\mathit{{q}}_{0}\right]&\!\!\!=&\!\!\!\!\!-{\frac{\partial \mathcal{L}_{\mathit{{q}}_{0}}(s)}{\partial s}}\bigg|_{s=0}, \quad \quad \quad \!\! 
\text{var}\left[\mathit{{q}}_{0}\right]={\frac{\partial^2 \mathcal{L}_{\mathit{{q}}_{0}}(s)}{\partial s^2}}\bigg|_{s=0} \!\!\!\!\!\!- \left( \mathbb{E}\left[\mathit{{q}}_{0}\right] \right)^2.
\end{eqnarray}

The proof to Lemma \ref{offDlemma} follows the same procedure with the only difference in that we replace $\mathit{{q}}_{0}$ in Eqs. \eqref{App1} and \eqref{Mformula} with $\mathit{\tilde{q}}_{0}$, and the Laplace transform, $ \mathcal{L}_{T}(s) $, 
in eq. \eqref{eq29a}
with the Laplace transform $ \mathcal{L}_{T^*}(s) $ of the composite random variable $ T^* \buildrel d \over = \mathit{L}_s \tilde{H}_1 \tilde{H}_2 $ where $\tilde{H}_1, \tilde{H}_2 \sim \mathcal{CN}(0,1)$ are independent and identically distributed complex normal random variables.
\vspace{-.15in}
\section*{Appendix C: Proof of Lemma \ref{Qfunction}}\vspace{-.05in}
We start by writing the expression for the EM algorithm Q-function using Proposition \ref{Proposition1} as
 \begin{eqnarray}
 Q({\theta}|{\theta}^{(m)})
 &=& \sum_{i=1}^n \mathbb{E}_{X_i|Y_i,\theta^{(m)}} \left[ \log f_{X_i} (x_i|\theta)\right] 
 =\sum_{i=1}^n \sum_{j=1}^2 p(Z_i=j|Y_i=y_i,\theta^{(m)}) \log f_{X_i} (x_i|\theta)\nonumber\\
 &=& \sum_{i=1}^n \sum_{j=1}^2 \gamma_{ij}^{(m)} \log w_j + \sum_{i=1}^n \gamma_{i1}^{(m)} \log \phi_1(y_i|\theta_1) + \sum_{i=1}^n \gamma_{i2}^{(m)} \log \phi_2(y_i|\theta_2). \label{The1eq}
 \end{eqnarray}
 Using the probability density functions of IG and IW distributions in Eqs. \eqref{f_IG} and \eqref{f_IW1} along with moment matching first, we can write $\log \phi_1(y_i|\lambda)$ and $\log \phi_2(y_i|c)$ as in Eqs. \eqref{logIG} and \eqref{logIW}, respectively. Then, substituting into Eq. \eqref{The1eq}, replacing $\sum_{i=1}^n \gamma_{ij}^{(m)}$ by $n_j^{(m)}$, 
 and dropping 
 constant terms 
 without affecting the EM algorithm, we get the final expression in Eq. \eqref{The1eq0}. 
\vspace{-.1in}
\section*{Appendix D: Proof of Theorem \ref{Theorem2b}}\vspace{-.05in}
We first solve for $w_j$ using the method of Lagrange multipliers as follows. 
\begin{eqnarray}
J(w,\nu)&=&\sum_{j=1}^2 n_{j}^{(m)} \log w_j +\nu \left( 1-\sum_{j=1}^2 w_{j} \right)\nonumber\\
\frac{\partial J}{\partial w_j}&=&\frac{n_j^{(m)}}{w_j}-\nu=0 \quad, j \in \{1,2\} \quad \Rightarrow \quad 
\nu =\frac{n_j^{(m)}}{w_j} \label{Lagrange}.
\end{eqnarray}
Substituting into the constraint $\sum_{j=1}^2 w_j=1$, we get the weight $w_j$ update as
\begin{eqnarray}
w_j^{(m+1)}&=&\frac{n_{j}^{(m)}}{\sum_{j=1}^k n_{j}^{(m)}}
=\frac{n_{j}^{(m)}}{\sum_{j=1}^k \sum_{i=1}^n \gamma_{ij}^{(m)}
} 
=\frac{n_{j}^{(m)}}{\sum_{i=1}^n 1}=\frac{n_{j}^{(m)}}{n} \quad, j \in \{1,2\}.\label{Wj}
\end{eqnarray}

Then, we can find the update of $\lambda_{IW}$ by simply taking the partial derivative of the $Q$-function w.r.t. $\lambda_{IW}$ as \vspace{-.1in} 
\begin{eqnarray}
\frac{\partial Q\left({{{\theta}}}|{{{\theta}}}^{(m)} \right)}{\partial {{\lambda_{}}}}&=& \frac{1}{2 {{\lambda_{}}}} n_1^{(m)} - \frac{1}{2\mu_Y^2} \sum_{i=1}^n \gamma_{i1}^{(m)} \frac{(y_i-\mu_Y)^2}{y_i}  =0. 
\end{eqnarray}
Solving for ${{\lambda_{}}}^{(m+1)}$ 
results in Eq. \eqref{Muj}.

Similarly, taking the partial derivative of the $Q$-function w.r.t. ${{c_{}}}$ to get
\begin{eqnarray}\label{Theo2eq2}
\frac{\partial Q\left({{{\theta}}}|{{{\theta}}}^{(m)} \right)}{\partial {{c_{}}}}&=& n_2^{(m)} \left[ \log \frac{\mu_Y}{\Gamma(1-\frac{1}{{{c_{}}}})} + \frac{1}{{{c_{}}}} \left(1-\psi(1-\frac{1}{{{c_{}}}}) \right) \right] - \sum_{i=1}^n \gamma_{i2}^{(m)} \log y_i \nonumber\\
&+& \left[ \frac{\mu_Y}{\Gamma(1-\frac{1}{{{c_{}}}})} \right]^{{{c_{}}}} \sum_{i=1}^n \gamma_{i2}^{(m)} y_i^{-{{c_{}}}} \log y_i 
%
-\frac{d}{d {{c_{}}}} \left( \left[ \frac{\mu_Y}{\Gamma(1-\frac{1}{{{c_{}}}})} \right]^{{{c_{}}}} \right) \sum_{i=1}^n \gamma_{i2}^{(m)} y_i^{-{{c_{}}}}. 
\end{eqnarray}
Equating to zero, we obtain the equation in \eqref{C_IW}.
\vspace{-.15in}
\section*{Appendix E: Fitted Functional Coefficients for the Mixture MLE Model}\label{ApxTables}\vspace{-.05in}
The fitted coefficient values in Table \ref{Table1} 
are obtained for the functional fits in \eqref{FitEq} by comparing our Mixture model against simulation data. These coefficients are estimated for a maximum transmission power of $P_s=30$ dBm. To estimate the mixture distribution parameters at a different transmission power $P_{new}$ dBm, we only need to update the $b, \lambda$ and $\mu$ parameters by scaling each of them by the corresponding transmit power scaling factor, i.e. $b_{new}=10^{0.1(P_{new}-P_s)} \times b$. Parameter $c$ stay unchanged with respect to the transmission power.
%
%
%
\begin{table}[h!]
{\setlength{\extrarowheight}{0.1cm}
\parbox{.5\linewidth}{
\centering
  \begin{tabular}{ | x{0.5cm} || x{1cm} | x{1cm} | x{1cm} | x{1cm} | }
    \hline \hline
    \diag{.09em}{.8cm}{$\text{  }\alpha$}{$\quad\quad \text{   coeff.}$} & $a_3$ & $a_2$ & $a_1$ & $a_0$  \\ \hline \hline
    2 & -111.99   & -589.09  	& -1021.1 & -873.55  \\ \hline
    3 & -41.663   & -214.46      & -443.27      & -433.55  \\ \hline
    3.5 & 13.4287 & -43.972      & -230.91      & -243.78  \\ \hline
    4 & 16.8309   & -17.117      & -137.71      & -149.15  \\ \hline
    5 & 5.7399    & -1.4256      & -40.481      & -122.15  \\ 
    \hline \hline
  \end{tabular}\vspace{-.1in}
  \subcaption{IG weight parameter $w_{1}(\times 10^3)$.}}\vspace{-.2in}
  \hfill 
  \parbox{.6\linewidth}{
\centering
  \begin{tabular}{ | x{0.5cm} || x{1cm} | x{1cm} | x{1cm} | x{1cm} | }
    \hline \hline
    \diag{.09em}{.8cm}{$\text{  }\alpha$}{$\quad\quad \text{   coeff.}$} & $a_3$ & $a_2$ & $a_1$ & $a_0$  \\ \hline \hline
    2 	& 7.5381  	& -23.171 		& -133.81 		& 685.25  \\ \hline
    3 	& 5.9903 	& -11.084 		& -81.932 		& 314.72  \\ \hline
    3.5 	& 7.0852 	& 8.8851 		& -105.21 		& 239.38  \\ \hline
    4 	& 15.616 	& 15.137 		& -142 			& 223.73  \\ \hline
    5 	& 11.126 	& -6.9276 		& -121.18 		& 224.88  \\ 
    \hline \hline
  \end{tabular}\vspace{-.1in}
  \subcaption{IW shape parameter ${{c_{}}}(\times 10^3)$.}}\vspace{-.03in}\\
  %
  %
  %
  \parbox{.5\linewidth}{
\centering
  \begin{tabular}{ | x{0.5cm} || x{1cm} | x{1cm} | x{1cm} | x{1cm} | }
    \hline \hline
    \diag{.09em}{.8cm}{$\text{  }\alpha$}{$\quad\quad \text{   coeff.}$} & $a_3$ & $a_2$ & $a_1$ & $a_0$  \\ \hline \hline
    2 	& -0.6594  	& 12.124 	& 17.277 	& 1134.1  \\ \hline
    3 	& 1.258  	& 47.252 	& 128.5  	& -2502.3  \\ \hline
    3.5 	& 1.604 		& 56.81 		& 117.1 		& -4102  \\ \hline
    4 	& 8.84 		& 56.04 		& 78.43 		& -5577  \\ \hline
    5 	& 13.99 		& 49.3 		& 65.3 		& -8338  \\ 
    \hline \hline
  \end{tabular}\vspace{-.1in}
  \subcaption{IG shape parameter ${{\lambda_{}}}(\times 10^3)$.}}\vspace{-.2in}
  \hfill
  \parbox{.6\linewidth}{
\centering
  \begin{tabular}{ | x{0.5cm} || x{1cm} | x{1cm} | x{1cm} | x{1cm} | }
    \hline \hline
    \diag{.09em}{.8cm}{$\text{  }\alpha$}{$\quad\quad \text{   coeff.}$} & $a_3$ & $a_2$ & $a_1$ & $a_0$  \\ \hline \hline
    2 	& 0.05993  	& 19.152 	& 44.478 	& -351.06  \\ \hline
    3 	& 0.25864  	& 105.43  	& 313.42 	& -2929.1  \\ \hline
    3.5 	& 0.82885	& 105.83 	& 311.6 		& -4186.3  \\ \hline
    4 	& 0.96339 	& 105.86 	& 311.24 	& -5398.9  \\ \hline
    5 	& 1.1855 	& 105.94 	& 310.68 	& -7752.1  \\ 
    \hline \hline
  \end{tabular}\vspace{-.1in}
  \subcaption{IG scale parameter ${{\mu_{}}}(\times 10^3)$.}}}
  \vspace{-.15in} \vspace{-.1in}
  \caption{Fitting coefficients for the mixture model parameters $w_1$, ${{b_{}}}$, ${{\lambda_{}}}$, and ${{\mu_{}}}$ at $\eta=1$.}\label{Table1} \vspace{-.3in}\vspace{-.3in}
\end{table}



\bibliographystyle{IEEEtran}
%

\bibliography{reflist}

\begin{thebibliography}{10}
\providecommand{\url}[1]{#1}
\csname url@samestyle\endcsname
\providecommand{\newblock}{\relax}
\providecommand{\bibinfo}[2]{#2}
\providecommand{\BIBentrySTDinterwordspacing}{\spaceskip=0pt\relax}
\providecommand{\BIBentryALTinterwordstretchfactor}{4}
\providecommand{\BIBentryALTinterwordspacing}{\spaceskip=\fontdimen2\font plus
\BIBentryALTinterwordstretchfactor\fontdimen3\font minus
  \fontdimen4\font\relax}
\providecommand{\BIBforeignlanguage}[2]{{%
\expandafter\ifx\csname l@#1\endcsname\relax
\typeout{** WARNING: IEEEtran.bst: No hyphenation pattern has been}%
\typeout{** loaded for the language `#1'. Using the pattern for}%
\typeout{** the default language instead.}%
\else
\language=\csname l@#1\endcsname
\fi
#2}}
\providecommand{\BIBdecl}{\relax}
\BIBdecl

\bibitem{6398884}
S.~Rajagopal, S.~Abu-Surra, and M.~Malmirchegini, ``Channel feasibility for
  outdoor non-line-of-sight mmwave mobile communication,'' in \emph{IEEE
  Vehicular Tech. Conf. (VTC Fall)}, Sept 2012, pp. 1--6.

\bibitem{mmWaveEnable}
W.~Roh, J.-Y. Seol, J.~Park, B.~Lee, J.~Lee, Y.~Kim, J.~Cho, K.~Cheun, and
  F.~Aryanfar, ``Millimeter-wave beamforming as an enabling technology for
  {{5G}} cellular communications: theoretical feasibility and prototype
  results,'' \emph{IEEE Comm. Magazine}, vol.~52, no.~2, pp. 106--113, February
  2014.

\bibitem{ref_2e}
F.~Boccardi, R.~Heath, A.~Lozano, T.~Marzetta, and P.~Popovski, ``Five
  disruptive technology directions for {{5G}},'' \emph{IEEE Comm. Magazine},
  vol.~52, no.~2, pp. 74--80, Feb. 2014.

\bibitem{Baimm14}
T.~Bai, V.~Desai, and R.~W. Heath, ``Millimeter wave cellular channel models
  for system evaluation,'' in \emph{IEEE ICNC}, 2014, pp. 178--182.

\bibitem{6894455}
T.~Bai, A.~Alkhateeb, and R.~Heath, ``Coverage and capacity of millimeter-wave
  cellular networks,'' \emph{IEEE Comm. Magazine}, vol.~52, no.~9, pp. 70--77,
  September 2014.

\bibitem{Int_Regime_7499308}
M.~Rebato, M.~Mezzavilla, S.~Rangan, F.~Boccardi, and M.~Zorzi, ``Understanding
  noise and interference regimes in 5g millimeter-wave cellular networks,'' in
  \emph{European Wireless 2016; 22th European Wireless Conference}, May 2016,
  pp. 1--5.

\bibitem{ref3}
A.~Giovanidis and F.~Baccelli, ``A stochastic geometry framework for analyzing
  pairwise-cooperative cellular networks,'' \emph{CoRR}, vol. abs/1305.6254,
  2013.

\bibitem{ref_6}
J.~Andrews, F.~Baccelli, and R.~Ganti, ``A tractable approach to coverage and
  rate in cellular networks,'' \emph{IEEE Trans. on Comm.}, vol.~59, no.~11,
  pp. 3122--3134, November 2011.

\bibitem{HElsawy6524460}
H.~ElSawy, E.~Hossain, and M.~Haenggi, ``Stochastic geometry for modeling,
  analysis, and design of multi-tier and cognitive cellular wireless networks:
  A survey,'' \emph{IEEE Comm. Surveys Tutorials}, vol.~15, no.~3, pp.
  996--1019, 2013.

\bibitem{heathHetero6515339}
R.~Heath, M.~Kountouris, and T.~Bai, ``Modeling heterogeneous network
  interference using {{Poisson}} point processes,'' \emph{IEEE Trans. on Signal
  Processing}, vol.~61, no.~16, pp. 4114--4126, Aug 2013.

\bibitem{Gaussian_Int_7541571}
S.~Ak, H.~Inaltekin, and H.~V. Poor, ``Gaussian approximation for the downlink
  interference in heterogeneous cellular networks,'' in \emph{2016 IEEE
  International Symposium on Information Theory (ISIT)}, July 2016, pp.
  1611--1615.

\bibitem{Gauss_Int_5549941}
M.~Aljuaid and H.~Yanikomeroglu, ``Investigating the gaussian convergence of
  the distribution of the aggregate interference power in large wireless
  networks,'' \emph{IEEE Transactions on Vehicular Technology}, vol.~59, no.~9,
  pp. 4418--4424, Nov 2010.

\bibitem{BaiH14}
T.~Bai and R.~Heath, ``Coverage and rate analysis for millimeter-wave cellular
  networks,'' \emph{IEEE Trans. on Wireless Comm.}, vol.~14, no.~2, pp.
  1100--1114, Feb 2015.

\bibitem{MyArxiv}
H.~Elkotby and M.~Vu, ``Uplink user-assisted relaying in cellular networks,''
  \emph{IEEE Trans. on Wireless Comm.}, vol.~14, no.~10, pp. 5468--5483, Oct
  2015.

\bibitem{ref_1r}
A.~Altieri, L.~Rey~Vega, P.~Piantanida, and C.~Galarza, ``Analysis of a
  cooperative strategy for a large decentralized wireless network,''
  \emph{IEEE/ACM Trans. on Networking}, vol.~22, no.~4, pp. 1039--1051, Aug
  2014.

\bibitem{ref2}
T.~Novlan, H.~Dhillon, and J.~Andrews, ``Analytical modeling of uplink cellular
  networks,'' \emph{IEEE Trans. on Wireless Comm.}, vol.~12, no.~6, pp.
  2669--2679, June 2013.

\bibitem{maccartney2013path}
G.~R. MacCartney, J.~Zhang, S.~Nie, and T.~S. Rappaport, ``Path loss models for
  {{5G}} millimeter wave propagation channels in urban microcells,'' in
  \emph{2013 IEEE Global Comm. Conf. (GLOBECOM)}, pp. 3948--3953.

\bibitem{mmWave4}
M.~R. Akdeniz, Y.~Liu, M.~K. Samimi, S.~Sun, S.~Rangan, T.~S. Rappaport, and
  E.~Erkip, ``Millimeter wave channel modeling and cellular capacity
  evaluation,'' \emph{IEEE Journal on Sel. Areas in Comm.}, vol.~32, no.~6, pp.
  1164--1179, 2014.

\bibitem{mmWave2}
M.~Samimi, T.~Rappaport, and G.~Maccartney, ``Probabilistic omnidirectional
  path loss models for millimeter-wave outdoor communications,'' \emph{IEEE
  Wireless Comm. Letters}, vol.~4, no.~4, pp. 357--360, Aug 2015.

\bibitem{mmWave6515173}
T.~Rappaport, S.~Sun, R.~Mayzus, H.~Zhao, Y.~Azar, K.~Wang, G.~Wong, J.~Schulz,
  M.~Samimi, and F.~Gutierrez, ``Millimeter wave mobile communications for
  {{5G}} cellular: It will work!'' \emph{IEEE Access}, vol.~1, pp. 335--349,
  2013.

\bibitem{mmWave6387266}
T.~Rappaport, F.~Gutierrez, E.~Ben-Dor, J.~Murdock, Y.~Qiao, and J.~Tamir,
  ``Broadband millimeter-wave propagation measurements and models using
  adaptive-beam antennas for outdoor urban cellular communications,''
  \emph{IEEE Trans. on Antennas and Propagation}, vol.~61, no.~4, pp.
  1850--1859, April 2013.

\bibitem{mmWaveSurvey15}
\BIBentryALTinterwordspacing
Y.~Niu, Y.~Li, D.~Jin, L.~Su, and A.~V. Vasilakos, ``A survey of millimeter
  wave (mmwave) communications for {{5G}}: Opportunities and challenges,''
  \emph{CoRR}, vol. abs/1502.07228, 2015. [Online]. Available:
  \url{http://arxiv.org/abs/1502.07228}
\BIBentrySTDinterwordspacing

\bibitem{turgut2016coverage}
E.~Turgut and M.~C. Gursoy, ``Coverage in heterogeneous downlink millimeter
  wave cellular networks,'' \emph{arXiv preprint arXiv:1608.01790}, 2016.

\bibitem{MM_6831093}
A.~Afzal and S.~A. Hassan, ``A stochastic geometry approach for outage analysis
  of ad hoc {{SISO}} networks in {{Rayleigh}} fading,'' in \emph{2013 IEEE
  Global Communications Conference (GLOBECOM)}, Dec 2013, pp. 336--341.

\bibitem{MM_6410048}
S.~Akoum and R.~W. Heath, ``Interference coordination: Random clustering and
  adaptive limited feedback,'' \emph{IEEE Transactions on Signal Processing},
  vol.~61, no.~7, pp. 1822--1834, April 2013.

\bibitem{MLE_signal_poor2013introduction}
H.~V. Poor, \emph{An introduction to signal detection and estimation}.\hskip
  1em plus 0.5em minus 0.4em\relax Springer Science \& Business Media, 2013.

\bibitem{CoverRelativeEntropy}
\BIBentryALTinterwordspacing
T.~M. Cover and J.~A. Thomas, \emph{Differential Entropy}.\hskip 1em plus 0.5em
  minus 0.4em\relax John Wiley \& Sons, Inc., 2005, pp. 243--259. [Online].
  Available: \url{http://dx.doi.org/10.1002/047174882X.ch8}
\BIBentrySTDinterwordspacing

\bibitem{KL_IG_agrawal2007efficacy}
R.~Agrawal \emph{et~al.}, ``On efficacy of {{Rayleigh-Inverse Gaussian}}
  distribution over {{K-distribution}} for wireless fading channels,''
  \emph{Wireless Communications and Mobile Computing}, vol.~7, no.~1, pp. 1--7,
  2007.

\bibitem{KL_6059452}
S.~Atapattu, C.~Tellambura, and H.~Jiang, ``A {{Mixture Gamma Distribution to
  Model the SNR of Wireless Channels}},'' \emph{IEEE Transactions on Wireless
  Communications}, vol.~10, no.~12, pp. 4193--4203, December 2011.

\bibitem{IG1tweedie1957statistical}
M.~C. Tweedie, ``Statistical {{Properties of Inverse Gaussian Distributions}}.
  {{I}}.'' \emph{The Annals of Math. Stat.}, pp. 362--377, 1957.

\bibitem{IG2tweedie1957statistical}
M.~C. Tweedie \emph{et~al.}, ``Statistical {{Properties of Inverse Gaussian
  Distributions}}. {{II}}.'' \emph{The Annals of Math. Stat.}, vol.~28, no.~3,
  pp. 696--705, 1957.

\bibitem{IG1}
R.~S.~C. J.~L.~Folks, ``The {{Inverse Gaussian Distribution and Its Statistical
  Application}} -- {{A Review}},'' \emph{Journal of the Royal Stat. Society.
  Series B (Methodological)}, vol.~40, no.~3, pp. 263--289, 1978.

\bibitem{IW1rinne2008weibull}
H.~Rinne, \emph{The Weibull distribution: a handbook}.\hskip 1em plus 0.5em
  minus 0.4em\relax CRC Press, 2008.

\bibitem{IW2sultan2014bayesian}
K.~Sultan, N.~Alsadat, and D.~Kundu, ``Bayesian and {{Maximum Likelihood
  Estimations of the Inverse Weibull parameters under progressive type-II
  censoring}},'' \emph{Journal of Stat. Computation and Sim.}, vol.~84, no.~10,
  pp. 2248--2265, 2014.

\bibitem{kountouris2014approximating}
M.~Kountouris and N.~Pappas, ``Approximating the interference distribution in
  large wireless networks,'' in \emph{2014 11th IEEE Int. Symp. on Wireless
  Comm. Systems (ISWCS)}, pp. 80--84.

\bibitem{dempster1977maximum}
A.~P. Dempster, N.~M. Laird, and D.~B. Rubin, ``Maximum {{Likelihood from
  Incomplete Data via the EM Algorithm}},'' \emph{Journal of the royal
  statistical society. Series B (methodological)}, pp. 1--38, 1977.

\bibitem{chen2010demystified}
Y.~Chen and M.~R. Gupta, ``{{EM}} demystified: An {{Expectation-Maximization
  Tutorial}},'' \emph{University of Washington}, 2010.

\end{thebibliography}

\end{document}